\documentclass[envcountsame]{llncs}
\usepackage{amsmath}

\usepackage{color}
\usepackage{amsfonts}
\usepackage{amssymb}
\usepackage{amsthm}
\usepackage{mathtools}
\theoremstyle{remark}
\newtheorem{exmpl}{Example}

\usepackage[linesnumbered,ruled]{algorithm2e}

\SetKwInOut{Input}{Input}
\SetKwInOut{Output}{Output}
\SetKw{Return}{return}
\SetKw{Error}{error}
\newcommand{\algoIfrule}{\SetKwIF{Ifrule}{ElseIfrule}{Else}{if}{then}{else if}{else}{endif}}

\DeclareMathOperator{\upper}{upper}
\DeclareMathOperator{\lowerop}{lower}
\DeclareMathOperator{\coeff}{coeff}
\DeclareMathOperator{\val}{val}
\DeclareMathOperator{\bound}{bound}
\DeclareMathOperator{\improves}{improves}
\DeclareMathOperator{\tight}{tight}
\DeclareMathOperator{\resolve}{resolve}
\DeclareMathOperator{\divpart}{div-part}
\DeclareMathOperator{\divderive}{div-derive}
\DeclareMathOperator{\cooper}{cooper}
\DeclareMathOperator{\lcm}{lcm}
\DeclareMathOperator{\divsolve}{div-solve}
\DeclareMathOperator{\prefix}{prefix}
\DeclareMathOperator{\topop}{top}

\DeclareMathOperator{\cores}{cores}
\DeclareMathOperator{\unblocked}{woSR}
\DeclareMathOperator{\level}{level}
\DeclareMathOperator{\Bsubseq}{B-subseq}
\DeclareMathOperator{\weight}{weight}

\DeclareMathOperator{\WCElimI}{CombDivs}
\DeclareMathOperator{\CombineDivs}{CombDivs}
\DeclareMathOperator{\vars}{vars}

\DeclareMathOperator{\chelp}{c}
\DeclareMathOperator{\lexhelp}{lex}
\DeclareMathOperator{\ishelp}{is}
\DeclareMathOperator{\guardedhelp}{guarded}

\newcommand{\cweight}{\weight_{\chelp}}

\newcommand{\CSArrow}{\Longrightarrow_{\mbox{\tiny{CS}}}}
\newcommand{\TightArrow}{\Longrightarrow_{\mbox{\tiny{tight}}}}

\newcommand{\TightSt}[3]{\langle #1, #2 \oplus #3 \rangle}

\newcommand{\SearchSt}[2]{\langle #1, #2 \rangle}
\newcommand{\ConflictSt}[3]{\langle #1, #2 \rangle \vdash #3}
\newcommand{\UnsatSt}{\mbox{unsat}}
\newcommand{\SatSt}[1]{\SearchSt{\upsilon[#1]}{\mbox{sat}}}
\newcommand{\decBnd}[3]{#1 #2 #3}
\newcommand{\propBnd}[4]{#1 #2_{#4} #3}
\newcommand{\ceilfun}[1]{\left\lceil #1 \right\rceil}
\newcommand{\floorfun}[1]{\left\lfloor #1 \right\rfloor}

\newcommand{\concseq}[1]{[\![ #1 ]\!]}
\newcommand{\emptyseq}{[\![ ]\!]}

\newcommand\FunUE[4]{#1 &#2 &#3 \quad &\mbox{ \textbf{if} } #4}

\newcommand\CSRuleName[1]{\shortintertext{\textbf{#1}}}

\newcommand\CSRuleUE[3]{\FunUE{#1}{\CSArrow}{#2}{#3}}

\newcommand\CSRuleME[3]{\FunUE{#1}{\CSArrow}{#2}{\left\lbrace \begin{array}{l} #3 \end{array} \right.}}

\newenvironment{CSRules}{\begin{array}{l l l l}}{\end{array}}

\newcommand\TightRule[4]{
\shortintertext{\textbf{#1}}
&#2 &\TightArrow &#3,\\
\shortintertext{#4}}

\newcommand\TightRuleWO[3]{
\shortintertext{\textbf{#1}}
&#2 &\TightArrow &#3}

\title{Linear Integer Arithmetic Revisited}
  \author{Martin Bromberger \and Thomas Sturm \and Christoph Weidenbach}
  \institute{Max Planck Institute for Informatics, Saarbr\"ucken, Germany   \email{\{mbromber,sturm,weidenb\}@mpi-inf.mpg.de}}
\begin{document}
  \maketitle

  \thispagestyle{plain}  
  
  \begin{abstract}
    We consider feasibility of linear integer programs in the context
    of verification systems such as SMT solvers or theorem provers.
    Although satisfiability of linear integer programs is decidable,
    many state-of-the-art solvers neglect termination in favor of
    efficiency.  It is challenging to design a solver that is both
    terminating and practically efficient.  Recent work by Jovanovi{\'c}
    and de Moura constitutes an important step into this direction.
    Their algorithm CUTSAT is sound, but does not terminate, in general.
    In this paper we extend their CUTSAT algorithm by refined inference
    rules, a new type of conflicting core, and a dedicated rule application
    strategy.  This leads to our algorithm CUTSAT++, which
    guarantees termination. 
\end{abstract}
  

  \section{Introduction}
\label{SE: Changes between CUTSAT and CUTSAT++}
\label{SE: Differences between CUTSAT and CUTSAT++}

Historically, feasibility of linear integer problems is a classical problem, which has been
addressed and thoroughly investigated by at least two independent research
lines: (i)~integer and mixed real integer linear programming for optimization~\cite{JungerLiebling:10a},
(ii)~first-order quantifier elimination and decision procedures for Presburger
 Arithmetic and corresponding 
complexity results~\cite{Presburger:29a,Cooper:72a,FischerRabin:74a,Oppen:78a,FerranteRackoff:75a,FerranteRackoff:79a,GathenSieveking:78a,Berman:77a,Berman:80a,Furer:82a,Gradel:87a,Weispfenning:90a,LasarukSturm:07a}.
We are interested in feasibility of linear integer problems, which we call simply \emph{problems}, in the context of
the combination of theories, as they occur, e.g., in the context of SMT solving
or theorem proving.  From this perspective, both these research lines address problems that are too general for our
purposes: with the former, the optimization aspects go considerably beyond pure
feasibility. The latter considers arbitrary Boolean combinations of constraints
and quantifier alternation or even parametric problems. 

Consequently, the SMT community has developed several interesting approaches on their
own \cite{CutProofs,Mathsat,SplittingOnDemand}.%,CutChase}. 
These solvers typically neglect termination and completeness in favour of efficiency.
More precisely, these approaches are based on a branch-and-bound strategy, where the rational relaxation of an integer problem is used to cut off and branch on integer solutions.
Together with the known a priori integer bounds~\cite{BoundingBox} for a problem this yields
a terminating and complete algorithm. However, these bounds are so large that
for many practical problems the resulting branch-and-bound search space cannot be
explored in reasonable time.
Hence, the a priori bounds are not integrated in the implementations of the approaches.

On these grounds, the recent work by Jovanovi{\'c} and de Moura~\cite{JovanovicM11,CutChase}, although itself not terminating,
 constitutes an important step towards an algorithm that is both efficient and terminating. The termination
does no longer rely on bounds that are a priori exponentially large in the occurring parameters.
Instead, it relies on structural properties of the problem, which are explored by their CUTSAT algorithm. 
The price for this result is an algorithm that is by far
more complicated than the above-mentioned branch-and-bound approach.
In particular, it has to consider divisibility constraints in addition to inequalities.

Our interest in an algorithm for integer constraints originates from a possible combination with superposition, e.g., see~\cite{FietzkeWeidenbach12}. 
In the superposition context integer constraints are part of the first-order clauses. 
Variables in constraints are typically unguarded, so that an efficient decision procedure for this case is a prerequisite for an
efficient combined procedure.

Our contribution is an extension and refinement of the CUTSAT algorithm, which we call CUTSAT++.
In contrast to CUTSAT, our CUTSAT++ generally terminates.
The basic idea of both algorithms is to reduce a problem containing unguarded integer
variables
%, e.g., the variables $x, y$ of Example~\ref{example: Conflict on initially unbounded constraints}, page~\pageref{example: Conflict on initially unbounded constraints},
to a problem containing only guarded variables.
%, e.g., the variable $z$ in the same example.
These unguarded variables are not eliminated. 
Instead, one explores the unguarded variables by adding 
constraints on smaller variables to the problem, with respect to a strict total ordering where all unguarded variables are larger than all guarded variables. 
After adding sufficiently many constraints, feasibility of the problem only depends on guarded variables.
Then a CDCL style algorithm tests for feasibility by employing exhaustive propagation.
The most sophisticated part is to turn an unguarded variable into a guarded variable. 
Quantifier elimination techniques, such as
Cooper elimination~\cite{Cooper:72a}, do so by removing the unguarded variable. 
In case of Cooper elimination, the price to pay is an exponentially growing Boolean
structure and exponentially growing coefficients (see Section~\ref{SE: Weak-Cooper-Elimination}). 
Since integer linear programming is NP-complete, all algorithms known today can not prevent such a kind of behavior, in general. 
Since Cooper elimination does not care about the concrete structure of a given problem, the exponential behavior is almost guaranteed.
The idea of both CUTSAT and CUTSAT++ is therefore to simulate a lazy variation of Cooper elimination.
This leaves space for model assumptions and simplification rules in order for the algorithm to adapt 
to the specific structure of a problem and hence to systematically avoid certain cases of the worst-case exponential behavior observed with Cooper elimination. 

This paper is and extended version of~\cite{BrombergerSW15}.
It is organized as follows. After fixing some notation in Section~\ref{SE: Preliminaries}, 
we present three examples for problems where
CUTSAT diverges. The divergence of CUTSAT can be fixed by respective refinements on the
original CUTSAT rules. However, in a fourth example the combination of our refinements results in a frozen state. 
Our conclusion is that CUTSAT lacks, in addition to our rule refinements, a third type of conflicting cores, which we call
\emph{diophantine conflicting core}. 
Theorem~\ref{lemma: Weak-Elim equivalence} in Section~\ref{SE: Weak-Cooper-Elimination} actually implies that any procedure that is based on what we call \emph{weak Cooper elimination} needs
to consider this type of conflicting core for completeness. %, Theorem~\ref{lemma: Weak-Elim equivalence}.
In Sections~\ref{SE: Strong Conflict Resolution}-\ref{SE: Termination and Completeness},
we refine the inference rules for the elimination of unguarded variables
on the basis of our results from Section~\ref{SE: Weak-Cooper-Elimination} and show their soundness, completeness, and termination.
We finally give conclusions and point at possible directions for future research.

  \section{Motivation}
\label{SE: Preliminaries}

We use \emph{variables} $x$, $y$, $z$, $k$, possibly with indices. Furthermore, we use \emph{integer constants} $a$, $b$, $c$, $d$, $e$, $l$, $v$, $u$, 
\emph{linear polynomials} $p$, $q$, $r$, $s$, and \emph{constraints} $I$, $J$ possibly with indices.
As input \emph{problems}, we consider finite sets of constraints $C$ corresponding to and sometimes used as conjunction over their elements.
Each constraint $I$ is either an inequality $a_n x_n + \ldots + a_1 x_1 + c \leq 0$ or 
a divisibility constraint $d \mid a_n x_n + \ldots + a_1 x_1 + c$. We denote $\coeff(I,x_i) = a_i \in \mathbb{Z}$. 
$\vars(C)$ denotes the set of variables occurring in $C$.
We sometimes write $C(x)$ in order to emphasise that $x \in \vars(C(x))$.
A problem $C$ is satisfiable if $\exists X: C$ holds, where $X=\vars(C)$. 
For true we denote $\top$ and for false we denote $\bot$.
Since $d \mid c x + s \equiv d \mid -c x -s$, we may assume that $c > 0$ for all $d \mid c x + s \in C$.
A variable $x$ is \emph{guarded} in a problem $C$ if $C$ contains constraints of the form $x - u_x \leq 0$ and $-x + l_x \leq 0$.
Otherwise, $x$ is \emph{unguarded} in $C$. Note that guarded variables are \emph{bounded} as defined in~\cite{CutChase} but not vice versa.
A constraint is \emph{guarded} if it contains only guarded variables. Otherwise, it is \emph{unguarded}.

Our algorithm CUTSAT++ aims at deciding whether or not a given problem $C$ is satisfiable.
It either ends in the state \emph{unsat}, or 
in a state $\langle \upsilon, \text{sat} \rangle$ where $\upsilon$ is a satisfiable assignment for $C$.
In order to reach one of those two final states, the algorithm produces \emph{lower bounds} $x \geq b$
and \emph{upper bounds} $x \leq b$ for the variables in $C$.
The produced bounds are stored in a sequence $M=\concseq{\gamma_1, \ldots, \gamma_n}$, which describes a partial model.
The empty sequence is denoted by $\emptyseq$. We use $\concseq{M,\gamma}$ and $\concseq{M_1,M_2}$ to denote the concatenation 
of a bound $\gamma$ at the end of $M$, and $M_2$ at the end of $M_1$, respectively.

By $\lowerop(x,M)=b$ and $\upper(x,M)=b$ we denote the value $b$ of the greatest lower bound $x \geq b$ and the least upper bound $x \leq b$ for a
variable $x$ in $M$, respectively~\cite{CutChase}.
If there is no lower (upper) bound for $x$ in $M$, then $\lowerop(x,M) = -\infty$ ($\upper(x,M) = \infty$).
The definitions of $\upper$ and $\lowerop$ are extended to polynomials as done in~\cite{CutChase}.

A state in CUTSAT++ is of the form $S = \SearchSt{M}{C}$ or $S = \ConflictSt{M}{C}{I}$,
or is one of the two \emph{final states} $\langle \upsilon, \text{sat} \rangle$, \emph{unsat}~\cite{CutChase}.
The \emph{initial-state} for a problem  $C$ is $\SearchSt{\emptyseq}{C}$.
For a state $S = \SearchSt{M}{C}(\vdash I)$, inequality $p \leq 0$ is a \emph{conflict} if $\lowerop(p,M) > 0$.
For a state $S = \SearchSt{M}{C}(\vdash I)$, divisibility constraint $d \mid ax + p$ is a \emph{conflict} if all variables 
in $p$ are fixed, and $d \nmid a b + \lowerop(p,M)$ for all $b$ with $\lowerop(x,M) \leq b \leq \upper(x,M)$.
In a state $S = \ConflictSt{M}{C}{I}$, the constraint $I$ is always a conflict.

The partial model $M$ of a state $\SearchSt{M}{C}(\vdash I)$ is complete if all variables $x$ in $C$ are \emph{fixed}, in the sense that $\upper(x,M) = \lowerop(x,M)$.
In this case, we define $\upsilon[M]$ as the assignment that assigns to every variable $x$ the value $\lowerop(x,M)$.
With $\val(p,M) = \lowerop(p,M)$ we denote the value assigned to a fixed polynomial $p$, i.e., $\val(p,M)$ is only defined if all variables occurring in $p$ are fixed in $M$.
A state is \emph{frozen} if it is not a final state and no rule is applicable.

Our CUTSAT++ algorithm is defined as a transition system consisting of the following rules:

\begin{align*}
\begin{CSRules}
\CSRuleName{Decide}
\CSRuleME{\SearchSt{M}{C}}{\SearchSt{\concseq{M, \decBnd{x}{\geq}{b}}}{C}}{\upper(x,M) \neq +\infty,\\ \lowerop(x,M) < b = \upper(x,M)} \\ \\[-2mm]
\CSRuleME{\SearchSt{M}{C}}{\SearchSt{\concseq{M, \decBnd{x}{\leq}{b}}}{C}}{\lowerop(x,M) \neq -\infty,\\ \lowerop(x,M) = b < \upper(x,M)} \\
\CSRuleName{Propagate}
\CSRuleME{\SearchSt{M}{C}}{\SearchSt{\concseq{M, \propBnd{x}{\geq}{b}{I}}}{C}}
{J \in C \mbox{ is an inequality}, \\ \coeff(J,x) < 0,\\ \improves(J,x,M),\\ b = \bound(J,x,M),\\ I = \tight(J,x,M)} \\ \\[-2mm]
\CSRuleME{\SearchSt{M}{C}}{\SearchSt{\concseq{M, \propBnd{x}{\leq}{b}{I}}}{C}}
{J \in C \mbox{ is an inequality}, \\ \coeff(J,x) > 0,\\ \improves(J,x,M),\\ b = \bound(J,x,M),\\ I = \tight(J,x,M)} \\
\CSRuleName{Propagate-Div}
\CSRuleME{\SearchSt{M}{C}}{\SearchSt{\concseq{M, \propBnd{x}{\geq}{c}{I}}}{C}}
{D = (d \mid ax+p) \in C, \val(p,M)=k,\\ b = \lowerop(x,M), d \nmid ab+k,\\
c = \bound(D,x,M), c \leq \upper(x,M), \\ 
I = \divderive(D,x,M)} \\ \\[-2mm]
\CSRuleME{\SearchSt{M}{C}}{\SearchSt{\concseq{M, \propBnd{x}{\leq}{c}{I}}}{C}}
{D = (d \mid ax+p) \in C, \val(p,M)=k,\\ b = \upper(x,M), d \nmid ab+k,\\
c = \bound(D,x,M), c \geq \lowerop(x,M),\\ 
I = \divderive(D,x,M)} \\
\CSRuleName{Conflict} \\[-1.5mm]
\CSRuleUE{\SearchSt{M}{C}}{\ConflictSt{M}{C}{p \leq 0}}{p \leq 0 \in C, \lowerop(p,M) > 0} \\[1.5mm]
\CSRuleName{Conflict-Div}
\CSRuleME{\SearchSt{M}{C}}{\ConflictSt{M}{C}{I}}
{J = (d \mid ax+p) \in C, \val(p,M)=k, \\ b = \lowerop(x,M), d \nmid ab+k,\\  
\bound(J,x,M) > \upper(x,M),\\ I = \divderive(J,x,M)\\} \\ \\[-2mm]
\CSRuleME{\SearchSt{M}{C}}{\ConflictSt{M}{C}{I}}
{J = (d \mid ax+p) \in C, \val(p,M)=k, \\ b = \upper(x,M), d \nmid ab+k,\\ 
\bound(J,x,M) < \lowerop(x,M),\\ I = \divderive(J,x,M)}
\end{CSRules}
\end{align*}

\begin{align*}
\begin{CSRules}
\CSRuleName{Unsat-Div}
\CSRuleME{\SearchSt{M}{C}}{\UnsatSt}{ d \mid a_1 x_1 + \ldots + a_n x_n + c \in C,\\ \gcd(d,a_1,\ldots,a_n) \nmid c} \\
\CSRuleName{Sat} \\[-1.5mm]
\CSRuleUE{\SearchSt{M}{C}}{\SatSt{M}}{\upsilon[M ] \text{ satisfies } C} \\[1.5mm]
\CSRuleName{Forget} \\[-1.5mm]
\CSRuleUE{\SearchSt{M}{C \cup \lbrace J \rbrace}}{\SearchSt{M}{C}}{C \vdash_{\mathbb{Z}} J,\text{and }J \not\in C} \\[1.5mm]
   \CSRuleName{Slack-Intro}
   \CSRuleME{\SearchSt{M}{C}}{\SearchSt{M}{C \cup C_s}}
   {\SearchSt{M}{C} \text{ is stuck},\\ \text{$x$ is stuck},\\ x_S \mbox{ is the slack-variable}, \\ C_s = \lbrace -x_S \leq 0, x-x_S \leq 0, \\ \qquad \quad -x-x_S \leq 0\rbrace}
\end{CSRules}
\end{align*}

\begin{align*}
\begin{CSRules}
\CSRuleName{Resolve} \\[-1.5mm]
\CSRuleUE{\ConflictSt{\concseq{M, \gamma}}{C}{I}}{\ConflictSt{M}{C}{\resolve(\gamma,I)}}{\gamma \text{ is an implied bound}} \\[1.5mm]
\CSRuleName{Skip-Decision}
\CSRuleME{\ConflictSt{\concseq{M, \gamma}}{C}{p \leq 0}}{\ConflictSt{M}{C}{p \leq 0}}{\gamma \text{ is a decided bound,}\\ \lowerop(p,M)>0} \\
\CSRuleName{Unsat} \\[-1.5mm]
\CSRuleUE{\ConflictSt{M}{C}{b \leq 0}}{\UnsatSt}{b > 0}  \\[1.5mm]
\CSRuleName{Backjump}
\CSRuleME{\ConflictSt{\concseq{M, \gamma, M'}}{C}{J}}{\SearchSt{\concseq{M,\propBnd{x}{\geq}{b}{I}}}{C}}
{\gamma \text{ is a decided bound,}\\ \coeff(J,x) < 0,\\ \improves(J,x,M),\\ I = \tight(J,x,M),\\ b = \bound(J,x,M)} \\ \\
\CSRuleME{\ConflictSt{\concseq{M, \gamma, M'}}{C}{J}}{\SearchSt{\concseq{M,\propBnd{x}{\leq}{b}{I}}}{C}}
{\gamma \text{ is a decided bound,}\\ \coeff(J,x) > 0,\\ \improves(J,x,M),\\ I = \tight(J,x,M),\\ b = \bound(J,x,M)} \\
\CSRuleName{Learn} \\[-1.5mm]
\CSRuleUE{\ConflictSt{M}{C}{I}}{\ConflictSt{M}{C \cup \{I\}}{I}}{I \not \in C}
\end{CSRules}
\end{align*}

\begin{align*}
\begin{CSRules}
   \CSRuleName{Solve-Div-Left}
   \CSRuleME{\SearchSt{M}{C}}{\SearchSt{M}{C'}}
    {\mbox{divisibility constraints } I_1,I_2 \in C, \\ x \mbox{ is top in } I_1 \mbox{ and } I_2, \\ x \mbox{ is unguarded},\\ \mbox{all other vars. in } I_1, I_2 \mbox{ are fixed}, \\
     (I'_1,I'_2) = \divsolve(x, I_1,I_2),\\ C' = C \setminus \lbrace I_1, I_2 \rbrace \cup \lbrace I'_1, I'_2 \rbrace, \\ I'_2 \mbox{ is not a conflict}}
   \CSRuleName{Solve-Div-Right}
   \CSRuleME{\SearchSt{M}{C}}{\SearchSt{M'}{C'}}
   {\mbox{divisibility constraints } I_1,I_2 \in C, \\ x \mbox{ is top in } I_1 \mbox{ and } I_2, \\ x \mbox{ is unguarded},\\ \mbox{all other vars. in } I_1, I_2 \mbox{ are fixed}, \\
    (I'_1,I'_2) = \divsolve(x, I_1,I_2),\\ C' = C \setminus \lbrace I_1, I_2 \rbrace \cup \lbrace I'_1, I'_2 \rbrace, \\ I'_2 \mbox{ is a conflict}, \\
    y  = \topop(I'_2), \\ M' = \prefix(M,y)}
   \CSRuleName{Resolve-Cooper}
   \CSRuleME{\SearchSt{M}{C}}{\SearchSt{M'}{C \cup R_k \cup R_c}}
   {(x,C') \mbox{ is a conflicting core}, \\ x \mbox{ is unguarded},\\ \mbox{ all } z \prec x \mbox{ are fixed and } C' \subseteq C, \\
    \mbox{if } J \in C \mbox{ is a conflict, then } \topop(J) \not \prec x, \\ \cooper(x,C') = (R_k,R_c),\\ y = \min_{I \in R_c}\{\topop(I)\},\\ M' = \prefix(M,y)}
\end{CSRules}
\end{align*}
CUTSAT~\cite{CutChase} includes all of the rules of CUTSAT++ except for the rules Solve-Div-Left, Solve-Div-Right, and Resolve-Cooper, which are explained in more detail in Section \ref{SE: Strong Conflict Resolution}.

Via applications of the rule Decide, CUTSAT++ adds \emph{decided bounds} $\decBnd{x}{\leq}{b}$ or $\decBnd{x}{\geq}{b}$ to the 
sequence $M$ in search-state $S$~\cite{CutChase}. 
A decided bound generally assigns a variable $x$ to $\lowerop(x,M)$ or $\upper(x,M)$.
Via applications of the propagation rules, CUTSAT++ adds \emph{propagated bounds} $\propBnd{x}{\geq}{b}{I}$ or $\propBnd{x}{\leq}{b}{I}$
to the sequence $M$, where $I$ is a generated constraint propagating the bound. To this end the
function $\bound(J,x,M)$ computes the strictest bound value $b$ and 
the function $\tight(J,x,M)$ computes the corresponding justification $I$ for constraint $J$ under the partial model $M$~\cite{CutChase}.
For an inequality $J$, $\bound(J,x,M)$ is defined as follows:\newline
\centerline{$
  \bound(ax+p \leq 0,x,M) = \left\lbrace \begin{array}{lll}
    -\left\lceil
      \frac{\lowerop(p,M)}{a}
    \right\rceil & \mbox{ if } a > 0 ,\\ \\
    -\left\lfloor \frac{\lowerop(p,M)}{a}\right\rfloor & \mbox{ if } a < 0 .
  \end{array} \right.
$}\newline \newline
Whenever $a>0$, $J$ propagates only upper bounds for $x$. Whenever $a<0$, $J$ propagates only lower bounds for $x$.
For a divisibility constraint $D$, $\bound(D,x,M)$ is defined as follows:\newline
\centerline{$
  \bound(d \mid ax+p,x,M) = \left\lbrace
  \begin{array}{lll}
     \left\lceil
      \frac{d \left\lceil \frac{a b + k}{d} \right\rceil - k}{a}
    \right\rceil & \mbox{ if } b = \lowerop(x,M) ,\\ \\
    \left\lfloor
      \frac{d \left\lfloor \frac{a b + k}{d} \right\rfloor - k}{a}
    \right\rfloor & \mbox{ if } b = \upper(x,M),
  \end{array} \right.
$}\newline \newline
where $a > 0$, $d > 0$, all variables in $p$ are fixed, and $\lowerop(p) = k$. Whenever we choose $b = \lowerop(x,M)$ in the above function, $D$ propagates a lower bound.
Whenever we choose $b = \upper(x,M)$ in the above function, $D$ propagates an upper bound. The bound value $\bound(D,x,M)$ for divisibility constraint $D$ is computed in such a way that CUTSAT++ never skips a satisfiable solution for $D$:

\begin{lemma}
\label{div-prop-lemma}
Let $D = d \mid ax+p$ be a divisibility constraint with $a > 0$. Let $\SearchSt{M}{C}$ be a state where the polynomial $p$ is fixed. Let $\bound(d \mid ax+p,x,M)$ denote a lower bound value. (1) Then it holds for all $e \in \{\lowerop(x,M), \ldots, \bound(d \mid ax+p,x,M)-1\}$ that $d \nmid a e + \lowerop(p,M)$.
Otherwise, $\bound(d \mid ax+p,x,M)$ denotes an upper bound value:
(2)  Then it holds for all $e \in \{\bound(d \mid ax+p,x,M,{\leq}) +1, \ldots, \upper(x,M)\}$ that $d \nmid a e + \lowerop(p,M)$.
\end{lemma}
\begin{proof}
We only prove the case that $\bound(d \mid ax+p,x,M)$ denotes a lower bound.
The proof for the second case is analogous.
Assume for a contradiction that there exists an $e \in \{\lowerop(x,M), \ldots, \bound(d \mid ax+p,x,M)-1\}$ such that $d \mid a e + \lowerop(p,M)$ holds.
Since $d \mid a e + \lowerop(p,M)$, it holds that
$ 
   \left\lceil \frac{a e + k}{d} \right\rceil = \frac{a e + k}{d} \mbox{ and } 
   \left\lceil
       \frac{d \left\lceil \frac{a e + k}{d} \right\rceil - k}{a}
   \right\rceil
   =
   \left\lceil
       \frac{d \frac{a e + k}{d} - k}{a}
   \right\rceil
   =
   e
$.
Since $\lowerop(x,M) \leq e$, it holds that
$
\left\lceil \frac{a \lowerop(x,M) + k}{d} \right\rceil \leq \left\lceil \frac{a e + k}{d} \right\rceil
$. Thence, \newline
\centerline{$
   \bound(d \mid ax+p,x,M,{\geq}) =
   \left\lceil
       \frac{d \left\lceil \frac{a b + k}{d} \right\rceil - k}{a}
   \right\rceil
   \leq
   \left\lceil
       \frac{d \left\lceil \frac{a e + k}{d} \right\rceil - k}{a}
   \right\rceil
   = e
$.}
Therefore, $e \notin \{\lowerop(x,M), \ldots, \bound(d \mid ax+p,x,M)-1\}$, which contradicts our initial assumption.
\end{proof}

The rules of the CUTSAT++ calculus are restricted in such a way that $M$ stays \emph{consistent}, i.e. $\lowerop(x,M) \leq \upper(x,M)$ for all variables $x \in X$.
CUTSAT++ also propagates only bounds that are more \emph{strict} than the current bound for the variable $x$, e.g., CUTSAT++ only propagates lower bound $x \geq b$ if $b > \lowerop(x,M)$.
This behaviour is expressed by the following predicate for inequalities $J = a x + p \leq 0$: \newline
\centerline{$
\improves(J,x,M) = \left\lbrace
\begin{array}{lll}
 &\lowerop(x,M) < \bound(J,x,M) \leq \upper(x,M)  &, \mbox{ if } a < 0,\\
 &\lowerop(x,M) \leq \bound(J,x,M) < \upper(x,M)  & ,\mbox{ if } a > 0.\\
\end{array} \right.
$}

The justifications annotated to the propagated bounds are necessary for a CDCL-like conflict resolution.
In CDCL, boolean resolution is used to combine the current conflict $C \vee l$ with a clause used for unit propagation $C' \vee \bar{l}$ to receive a new conflict $C \vee C'$ without literals $l$ or $\bar{l}$.
For CUTSAT++ the function $\resolve$ fulfils a similar purpose: \newline
\centerline{$
\resolve(\propBnd{x}{\bowtie}{b}{cx+q \leq 0}, ax+p \leq 0)
 = 
\left\{\begin{array}{ll}
|a|q + |c|p \leq 0 &, \text{ if } a \cdot c < 0 ,\\
ax+p\leq0 &, \text{ otherwise }.
\end{array} \right.
$}
Whenever $J = ax+p \leq 0$ is a conflict in state $\ConflictSt{\concseq{M, \gamma}}{C}{J}$, where $C \vdash_{\mathbb{Z}} J$ and $\gamma = \propBnd{x}{\bowtie}{b}{cx+q \leq 0}$, $J' = \resolve(\propBnd{x}{\bowtie}{b}{cx+q \leq 0}, ax+p \leq 0)$ is also a conflict in state $\ConflictSt{M}{C}{J'}$ and $C \vdash_{\mathbb{Z}} J'$~\cite{CutChase}.
The function $\resolve$ only results in a new conflict because CUTSAT++ requires that the justification $I$ of bound $\propBnd{x}{\bowtie}{b}{I}$, with ${\bowtie} \in \{\leq,\geq\}$, in state $\SearchSt{M}{C}(\vdash J)$ fulfils the following conditions:
Firstly, $I$ is an inequality and $C \vdash_{\mathbb{Z}} I$.
Secondly, if ${\bowtie} = {\leq}$, then $\coeff(I,x) = 1$.
If ${\bowtie} = {\geq}$, then $\coeff(I,x) = -1$.
Finally, $\bound(I,x,M) \bowtie b$, i.e. the justification $I$ implies at least a bound as strong as $\propBnd{x}{\bowtie}{b}{I}$.

The function $\tight(J,x,M) = I$ defined by the set of rules in Figure~\ref{fig: tight rules} calculates a justification $I$ in variable $x$ for the bound $\propBnd{x}{\bowtie}{b}{I}$ propagated from inequality $J = \pm ax + p \leq 0$, where $b = \bound(J,x,M)$.
A state in this rule system is a pair \newline
\centerline{$
\TightSt{M'}{\pm ax + as}{r},
$}
where $a > 0$, $s$ and $r$ are polynomials, and $M'$ is a prefix of the initial $M$, i.e. $M = \concseq{ M', M''}$.
The initial state for $\tight(\pm ax + p \leq 0,x,M)$ is $\TightSt{M}{\pm ax}{p}$.
The first goal is to produce an inequality where all coefficients are divisible by $a = \coeff(J,x)$.
To this end, we apply the rules (Fig. \ref{fig: tight rules}) until the polynomial on the right side of ${\oplus}$ becomes empty.
In this case, all coefficients are divisible by $a = \coeff(J,x)$ and we derive the justification $I$ with the final rule Round.

\begin{figure}[t]
\begin{align*}
\TightRule{Consume}{\TightSt{ M}{ \pm ax + as }{ aky + r }}{\TightSt{ M}{ \pm ax + as + aky }{ r }}{where $x \neq y$.}
\TightRule{Resolve-Implied}{\TightSt{ \concseq{M, \gamma}}{ \pm ax + as }{ p }}{\TightSt{ M}{ \pm ax + as }{ q }}{where $\gamma$ is a propagated bound and $q \leq 0 = \resolve(\gamma,p \leq 0)$.}
\TightRule{Decide-Lower}{\TightSt{ \concseq{M, \decBnd{y}{\geq}{b}}}{ \pm ax + as }{ cy + r }}{\TightSt{ M}{ \pm ax + as + aky }{ r + (ak-c)q }}{where $\propBnd{y}{\leq}{b}{I}$ in $M$, with $I = y + q \leq 0$, and $k = \ceilfun{ \frac{c}{a}}$.}
\TightRule{Decide-Lower-Neg}{\TightSt{ \concseq{M, \decBnd{y}{\geq}{b}}}{ \pm ax + as }{ cy + r }}{\TightSt{ M}{ \pm ax + as }{ cq + r }}{where $\propBnd{y}{\leq}{b}{I}$ in $M$, with $I = y - q \leq 0$, and $c < 0$.}
\TightRule{Decide-Upper}{\TightSt{ \concseq{M, \decBnd{y}{\leq}{b}}}{ \pm ax + as }{ cy + r }}{\TightSt{ M}{ \pm ax + as + aky }{ r + (c-ak)q }}{where $\propBnd{y}{\geq}{b}{I}$ in $M$, with $I = -y + q \leq 0$, and $k = \floorfun{ \frac{c}{a}}$.}
\TightRule{Decide-Upper-Pos}{\TightSt{ \concseq{M, \decBnd{y}{\leq}{b}}}{ \pm ax + as }{ cy + r }}{\TightSt{ M}{ \pm ax + as }{ cq + r }}{where $\propBnd{y}{\geq}{b}{I}$ in $M$, with $I = -y + q \leq 0$, and $c > 0$.}
\TightRuleWO{Round (and terminate)}{\TightSt{ M}{ \pm ax + as }{ b }}{ \pm x + s + \ceilfun{ \frac{b}{a} } \leq 0 }
\end{align*}
\caption{Rule system that derives tightly propagating inequalities~\cite{CutChase}}
\label{fig: tight rules}
\end{figure}

The function $\divderive(D,x,M)$ calculates a justification $I$ in variable $x$ for the bound $\propBnd{x}{\bowtie}{b}{I}$ propagated from the divisibility constraint $D = d \mid ax + p$, where $b = \bound(D,x,M)$.
The justification computed by $\divderive$ is derived from a set of inequalities describing the propagated bound value.
For instance, let us look at the lower bound value: \newline
\centerline{$
\bound(D,x,M) = \left\lceil \frac{d \left\lceil \frac{a b + k}{d} \right\rceil - k}{a}\right\rceil,
$}
where $D = d \mid ax + p$, $a > 0$, $d > 0$, $p$ is fixed, $k = \lowerop(p,M)$, and $b = \lowerop(x,M)$.
The sub-term $c = \left\lceil \frac{a b + k}{d} \right\rceil$ is equal to the bound value computable from the \emph{diophantine representation} $d z = ax + p$ of the divisibility constraint $D$: \newline
\centerline{$
\bound(-dz +ax +p \leq 0,z,M) = \left\lceil \frac{a b + k}{d} \right\rceil.
$}
Notice that $z$ is a variable not occurring in the problem, and only introduced for the above calculation.
The fitting tightly propagating inequality for the sub-terms are abbreviated with $\divpart(D,x,M)$: \newline
\centerline{$
  \divpart(D,x,M) = \left\lbrace
  \begin{array}{lll}
    \tight(-dz +ax +p \leq 0,x,M) &, \mbox{ if } b = \lowerop(x,M), \\
    \tight(dz -ax -p \leq 0,x,M) &, \mbox{ if } b = \upper(x,M).
  \end{array} \right.
$}
For $\divpart$, we forbid $\tight$ to apply the Consume rule to the variables $x$ and $z$.
The restriction to Consume guarantees that the inequality $\pm z + r \leq 0 = \divpart(D,x,M)$ does not contain $x$.
Given $I_2 = -z + r \leq 0 = \divpart(D,x,M)$ and $I_1 = d z - ax - p \leq 0$, we use $\resolve(\propBnd{x}{\geq}{c}{I_2},I_1) = -ax +dr -p \leq 0 = I_3$ to receive the inequality that computes the complete lower bound: \newline
\centerline{$
  \bound(I_3,x,M) = \left\lceil \frac{d \lowerop(r,M) - k}{a}\right\rceil \geq  \left\lceil \frac{d \left\lceil \frac{a b + k}{d} \right\rceil - k}{a}\right\rceil.
$}
Finally, we compute the tightly propagating inequality for a divisibility constraint $D = d \mid ax + p$ with the function $\divderive(D,x,M)$: \newline
\centerline{$
  \divderive(D,x,M) = \left\lbrace \begin{array}{lll}
     I'_3 & \mbox{ if } \left\lbrace \begin{array}{l}
                                                                                         b = \lowerop(x,M), \\
                                                                                         I_2 = -z + r \leq 0 = \divpart(D,x,M,), \\
                                                                                         I'_3 = \tight(-ax +dr -p \leq 0,x,M), \\
                                                                                         \end{array} \right. \\
    I'_3 & \mbox{ if } \left\lbrace \begin{array}{l}
                                                                                         b = \upper(x,M), \\
                                                                                         I_2 = z + r \leq 0 = \divpart(D,x,M), \\
                                                                                         I'_3 = \tight(ax +dr +p \leq 0,x,M). \\
                                                                                         \end{array} \right.
  \end{array} \right.
$}\newline

The rule Slack-Intro is necessary to prevent a special type of frozen state called stuck state.
A variable $x$ is called \emph{stuck} in state $S = \SearchSt{M}{C}$ if $M$ contains no bounds for $x$
and there is no inequality $I = a x + p \leq 0 \in C$ that propagates a bound for $x$~\cite{CutChase}.
Variables $x$ with a constraint of the form $\pm x-b \leq 0 \in C$ are never stuck, as CUTSAT++ is able to propagate at least one bound for $x$, i.e., either $x \geq -b$ or $x \leq b$.
A state $S$ is a \emph{stuck state} if all unfixed variables $x$ are stuck and if the rules Sat, Unsat-Div, Conflict, and Conflict-Div are not applicable~\cite{CutChase}. 
In a stuck state, Slack-Intro is applicable and one of the previously stuck variables $x$ is turned unstuck by the constraints added to the problem.
As recommended in~\cite{CutChase}, CUTSAT++ uses the same slack variable for all Slack-Intro applications.

\medskip
We are now going to discuss three examples where CUTSAT diverges. The first one shows that, CUTSAT can apply Conflict and Conflict-Div 
infinitely often to constraints containing unguarded variables.

\begin{exmpl}
Let\newline
\centerline{$C := \lbrace \underbrace{-x \leq 0}_{I_x}, \underbrace{-y \leq 0}_{I_y}, \underbrace{-z \leq 0}_{I_{\mathrm{z1}}}, \underbrace{z \leq 0}_{I_{\mathrm{z2}}}, \underbrace{z + 1\leq 0}_{I_{\mathrm{z3}}}, \underbrace{ 1 - x + y \leq 0 }_{J_1}, \underbrace{x - y \leq 0}_{J_2} \rbrace$}
be a problem.
Let $S_i = \SearchSt{M_i}{C}$ for $i \in \mathbb{N}$ be a series of states with:\newline
\centerline{$\begin{array}{l l}
  M_0 &:= \concseq{\propBnd{x}{\geq}{0}{I_x}, \propBnd{y}{\geq}{0}{I_y}, \propBnd{z}{\geq}{0}{I_{\mathrm{z1}}}, \propBnd{z}{\leq}{0}{I_{\mathrm{z2}}} }, \\
  M_{i + 1} &:= \concseq{M_{i}, \propBnd{x}{\geq}{i + 1}{J_1}, \propBnd{y}{\geq}{i + 1}{J_2} }.
\end{array}$}
Let the variable order be given by $z \prec y \prec x $.
CUTSAT with a two-layered strategy~\cite{CutChase} has diverging runs starting 
in state $S'_0 = \SearchSt{\emptyseq}{C}$.
Let CUTSAT traverse the states $S'_0$, $S_0$, $S_1$, $S_2$, $\ldots$ in the following fashion:
CUTSAT reaches state $S_0$ from state $S'_0$ after propagating the constraints $I_x$, $I_y$, $I_{\mathrm{z1}}$, and $I_{\mathrm{z2}}$.
CUTSAT reaches state $S_{i+1}$ from state $S_i$ after:
\begin{itemize}
\item  fixing $x$ to $i$ with a Decision $\gamma^x_d := \decBnd{x}{\leq}{i}$; $M^1_i := \concseq{M_i,\gamma^x_d}$ and $S^1_i := \SearchSt{M^1_i}{C}$
\item  applying Conflict to Constraint $J_1$ because $\lowerop(1 - x + y,M^1_i) > 0$; $M^2_i := M^1_i$ and $S^2_i := \ConflictSt{M^2_i}{C}{J_1}$
\item  undoing the decided bound $\gamma^x_d$ by applying Backjump because the predicate $\improves(J_1,x,M_i,{\geq})$ evaluates to true.
       The result is the exchange of $\gamma^x_d$ with the bound $\gamma^x = \propBnd{x}{\geq}{i + 1}{J_1}$; 
       $M^3_i := \concseq{M_i, \gamma^x}$ and $S^3_i := \SearchSt{M^3_i}{C}$
\item  fixing $y$ to $i$ with a Decision $\gamma^y_d := \decBnd{y}{\leq}{i}$; $M^4_i := \concseq{M^3_i,\gamma^y_d}$ and $S^4_i := \SearchSt{M^4_i}{C}$
\item  applying Conflict to Constraint $J_2$ because $\lowerop(x - y ,M^4_i) > 0$; $M^5_i := M^4_i$ and $S^5_i := \ConflictSt{M^5_i}{C}{J_1}$
\item  undoing the decided bound $\gamma^y_d$ by applying Backjump because the predicate  $\improves(J_2,y,M^3_i,{\geq})$ evaluates to true. 
       The result is the exchange of $\gamma^y_d$ with the bound $\gamma^y = \propBnd{y}{\geq}{i + 1}{J_2}$;
       $M^6_i := \concseq{M^3_i, \gamma^y}$ and $S_{i+1} = S^6_i := \SearchSt{M^6_i}{C}$
\end{itemize}
Since $\{I_{\mathrm{z1}},I_{\mathrm{z3}}\}$ is a conflicting core, the variable $z$ is the minimal conflicting variable in the 
states $S_i$, $S^1_i$, and $S^4_i$.
Since $I_{\mathrm{z1}}$ and $I_{\mathrm{z2}}$ bound $z$, the conflicting core is also guarded.
Therefore, Resolve-Cooper as defined in \cite{CutChase} is not applicable, which in turn implies that Conflict 
is applicable.
\label{example: Conflict on initially unbounded constraints}
\end{exmpl}
A straightforward fix to example~\ref{example: Conflict on initially unbounded constraints} is to limit the application of the Conflict and Conflict-Div rules
to guarded constraints.
Our second example shows, that CUTSAT can still diverge by infinitely many applications of the Solve-Div rule~\cite{CutChase}.
\begin{exmpl}
\label{example: Solve-Div top-var restriction}
Let $d_i$ be the sequence with $d_0 = 2$ and $d_{k+1} := {d_k}^2$ for 
$k \in \mathbb{N}$, let $C_0 = \{4 \mid 2x + 2y, 2 \mid x + z\}$ be a problem, 
and let $S_0 = \SearchSt{\emptyseq}{C_0}$ be the initial CUTSAT state.
Let the variable order  be given by $x \prec y \prec z$.
Then CUTSAT has divergent runs 
$S_0 \CSArrow S_1 \CSArrow S_2 \CSArrow \ldots$. For instance, let CUTSAT apply the Solve-Div rule 
whenever applicable. % (see~\cite{CutChase} pp.~98-99).
By an inductive argument, Solve-Div is applicable in every 
state $S_n = \SearchSt{\emptyseq}{C_n}$, and the constraint set $C_n$ has the following form:\newline
\centerline{$
C_n = \left\{\begin{array}{l l}
 \{ 2 d_{n} \mid d_n x + d_n y, d_n \mid \frac{d_n}{2} y - \frac{d_n}{2} z\} & \mbox{ if } n \mbox{ is odd}, \\
\{ 2 d_{n} \mid d_n x + d_n y, d_n \mid \frac{d_n}{2} x + \frac{d_n}{2} z\} & \mbox{ if } n \mbox{ is even}. \\
\end{array}\right.$}
Therefore, CUTSAT applies Solve-Div infinitely often and diverges.
\end{exmpl}
A straightforward fix to example~\ref{example: Solve-Div top-var restriction} is to limit the application of Solve-Div to maximal variables in the variable order ${\prec}$.
Our third example shows, that CUTSAT can apply Conflict and Conflict-Div~\cite{CutChase} infinitely often. %(see~\cite{CutChase} p.~101) 
The example~\ref{example: CC top-var restriction} differs from example~\ref{example: Conflict on initially unbounded constraints} in that the conflicting core contains also unguarded variables.
\begin{exmpl}
\label{example: CC top-var restriction}
Let\newline
\centerline{$C := \lbrace \underbrace{-x \leq 0}_{I_x}, \underbrace{-y \leq 0}_{I_y}, \underbrace{-z \leq 0}_{I_{\mathrm{z1}}}, \underbrace{z \leq 0}_{I_{\mathrm{z2}}}, \underbrace{ 1 - x + y + z \leq 0 }_{J_1}, \underbrace{x - y - z \leq 0}_{J_2} \rbrace\}$}
be a problem.
Let $S_i = \SearchSt{M_i}{C}$ for $i \in \mathbb{N}$ be a series of states with:\newline
\centerline{$
\begin{array}{l l}
  M_0 &:= \concseq{\propBnd{x}{\geq}{0}{I_x}, \propBnd{y}{\geq}{0}{I_y}, \propBnd{z}{\geq}{0}{I_{\mathrm{z1}}}, \propBnd{z}{\leq}{0}{I_{\mathrm{z2}}} },\\
  M_{i + 1} &:= \concseq{M_{i}, \propBnd{x}{\geq}{i + 1}{J_1}, \propBnd{y}{\geq}{i + 1}{J_2} }.
\end{array}$
}
Let the variable order be given by $z \prec x \prec y$.
CUTSAT has diverging runs starting in state $S'_0 = \SearchSt{\emptyseq}{C}$.
For instance, let CUTSAT traverse the states $S'_0$, $S_0$, $S_1$, $S_2$, $\ldots$ in the following fashion:
CUTSAT reaches state $S_0$ from state $S'_0$ after propagating the constraints $I_x$, $I_y$, $I_{\mathrm{z2}}$ and $I_{\mathrm{z2}}$.
CUTSAT reaches state $S_{i+1}$ from state $S_i$ after:
\begin{itemize}
\item  fixing $x$ to $i$ and $y$ to $i$ with Decisions $\gamma^x_d := \decBnd{x}{\leq}{i}$ and $\gamma^y_d := \decBnd{y}{\leq}{i}$; $M^1_i := \concseq{M_i,\gamma^x_d,\gamma^y_d}$ and $S^1_i := \SearchSt{M^1_i}{C}$
\item  applying Conflict to Constraint $J_1$ because $\lowerop(1 - x + y + z,M^1_i) > 0$; $M^2_i := M^1_i$ and $S^2_i := \ConflictSt{M^2_i}{C}{J_1}$
\item  undoing the decided bounds $\gamma^y_d$  and $\gamma^x_d$ by applying first Skip-Decision and then Backjump.
       The result is the sequence $M^3_i := \concseq{M_i, \gamma^x}$ and the state $S^3_i := \SearchSt{M^3_i}{C}$, where $\gamma^x = \propBnd{x}{\geq}{i + 1}{J_1}$; 
\item  fixing $y$ to $i$ and $x$ to $i+1$ with Decisions $\gamma^y_d := \decBnd{y}{\leq}{i}$ and $\gamma^x_d := \decBnd{y}{\leq}{i+1}$; $M^4_i := \concseq{M^3_i,\gamma^y_d,\gamma^x_d}$ and $S^4_i := \SearchSt{M^4_i}{C}$
\item  applying Conflict to Constraint $J_2$ because $\lowerop(x - y -z ,M^4_i) > 0$; $M^5_i := M^4_i$ and $S^5_i := \ConflictSt{M^5_i}{C}{J_1}$
\item  undoing the decided bounds $\gamma^x_d$  and $\gamma^y_d$ by applying first Skip-Decision and then Backjump. 
       The result is the sequence $M^6_i := \concseq{M^3_i, \gamma^y}$ and the state $S_{i+1} = S^6_i := \SearchSt{M^6_i}{C}$, where $\gamma^y = \propBnd{y}{\geq}{i + 1}{J_2}$.
\end{itemize}
Notice that the conflicting core $\{J_1,J_2\}$ in states $S^1_i$ and $S^4_i$ is bounded in \cite{CutChase}, which admits the application of Conflict.
\end{exmpl}

For example~\ref{example: CC top-var restriction}, applying the fix suggested for example~\ref{example: Conflict on initially unbounded constraints} results in a frozen state.
Here, a straightforward fix is to change the definition of conflicting cores to cover only those cores where the conflicting variable
is the maximal variable.\footnote{The restrictions to maximal variables in the definition of the conflicting core and to 
the Solve-Div rule were both confirmed as missing but necessary in a private communication with Jovanovi\'c.}

The fixes for our examples suggested above are restrictions of CUTSAT which have the consequence that
 Conflict(-Div) cannot be applied to unguarded constraints, Solve-Div is only applicable for the elimination of the maximal variable, and the conflicting variable $x$ is the maximal variable in the associated conflicting core $C'$.
However, our next and final example shows that these restrictions lead to frozen states.

\begin{exmpl} \label{example:frozenstate}
Let CUTSAT include restrictions to maximal variables in the definition of conflicting cores, and in the Solve-Div rule as described above.
Let there be additional restrictions in CUTSAT to the rules Conflict and Conflict-Div such that these rules are only applicable to conflict constraints $I$ where $I$ contains no unguarded variable.
Let\newline 
\centerline{$C := \lbrace \underbrace{-x \leq 0}_{I_{\mathrm{x1}}}, \underbrace{x - 1\leq 0}_{I_{\mathrm{x2}}}, \underbrace{-y \leq 0}_{I_y}, 
\underbrace{6 \mid 4 y + x}_{J} \rbrace$} 
be a problem.
Let $M := \concseq{\propBnd{x}{\geq}{0}{I_{\mathrm{x1}}}, \propBnd{x}{\leq}{1}{I_{\mathrm{x2}}}, \propBnd{y}{\geq}{0}{I_y}, \decBnd{x}{\geq}{1}, \decBnd{y}{\leq}{0}}$ be a bound sequence.
Let the variable order be given by $x \prec y$.
CUTSAT has a run starting in state $S'_0 = \SearchSt{\emptyseq}{C}$ that ends in the frozen state $S = \SearchSt{M}{C}$.
Let CUTSAT propagate $I_{\mathrm{x1}}$, $I_{\mathrm{x2}}$, $I_y$, and fix $x$ to $1$ and $y$ to $0$ with two Decisions.
Through these Decisions, the constraint $J$ is a conflict.
Since $y$ is unguarded, CUTSAT cannot apply the rule Conflict-Div.
Furthermore, \cite{CutChase} has defined conflicting cores as either interval or divisibility conflicting cores. % (see~\cite{CutChase} p.~101).
The state $S$ contains neither an interval or a divisibility conflicting core.
Therefore, CUTSAT cannot apply the rule Resolve-Cooper.
The remaining rules are also not applicable because all variables are fixed and there 
is only one divisibility constraint.
Without the before introduced restrictions to the rules Conflict(-Div), Solve-Div, CUTSAT diverges on the example.
\end{exmpl}

  \section{Weak Cooper Elimination}
\label{SE: Weak-Cooper-Elimination}
\label{SSE: Weak-Cooper-Elimination}

In order to fix the frozen state of Example~\ref{example:frozenstate} in the previous section, we are going to
introduce in Section \ref{SE: Strong Conflict Resolution} a new conflicting core, which we call \emph{diophantine conflicting core}. 
For understanding diophantine conflicting cores, as well as further modifications to be made, it is helpful to understand the connection between CUTSAT and a variant of Cooper's quantifier elimination procedure~\cite{Cooper:72a}.

The original \emph{Cooper elimination} takes a variable $x$, a problem $C(x)$, and produces a disjunction of problems, equivalent
to $\exists x:C(x)$:\newline
\centerline{$\exists x: C(x) \equiv \bigvee\limits_{0 \leq k < m} C_{-\infty}(k) \vee \bigvee\limits_{- a x + p \leq 0 \in C} 
\bigvee\limits_{0 \leq k < a \cdot m} \left[ a \mid p + k \wedge C\left( \frac{p+k}{a} \right) \right],$}
where $a > 0$, $m = \lcm\{d \in \mathbb{Z} : (d \mid a x + p) \in C\}$.
If there exists no constraint of the form $-a x + p \leq 0 \in C$, then $C_{-\infty}(x) = \{(d \mid a x + p) \in C\}$ .
Otherwise, $C_{-\infty}(x) = \bot$.
One application of Cooper elimination results in a disjunction of quadratically many problems out of a single problem. 
Iteration causes an exponential increase in the coefficients due to the multiplication with $a$ because division is not part of the language.

\emph{Weak Cooper elimination} is a variant of Cooper elimination that is very helpful to understand problems around CUTSAT. 
The idea is, instead of building a disjunction over all potential solutions
for $x$, to add additional guarded variables and constraints without $x$ that guarantee the existence of a solution for $x$. 
We assume here that $C(x)$ contains only one divisibility constraint for $x$. 
If not, exhaustive application of $\divsolve$ to divisibility constraints for
$x$ removes all constraints except one: 
$\divsolve(x, d_1 \mid a_1 x + p_1, d_2 \mid a_2 x + p_2) =  (d_1 d_2 \mid d x + c_1 d_2 p_1 + c_2 d_1 p_2, d \mid -a_1 p_2 + a_2 p_1)$,
where $d = \gcd(a_1 d_2, a_2 d_1)$, and $c_1$ and $c_2$ are integers such that $c_1 a_1 d_2 + c_2 a_2 d_1 = d$~\cite{Cooper:72a,CutChase}. 
Now weak Cooper elimination takes a variable $x$, a problem $C(x)$ and produces a new problem by replacing $\exists x:C(x)$ with:\newline
\centerline{$\exists K: \left(\{I \in C(x) : \coeff(I,x) = 0 \} \cup 
                                  \{\gcd(c,d) \mid s\} \cup \underset{k \in K}{\bigcup} R_k \right)$,}
where $d \mid c x + s \in C(x)$, $k \in K$ is a newly introduced variable for every pair of 
constraints $-a x + p \leq 0 \in C(x)$ and $b x - q \leq 0 \in C(x)$,\newline
\centerline{\label{resolvenpage}$R_k = \{-k \leq 0, k - m \leq 0, bp-aq+bk \leq 0, a \mid k+p, \,ad\mid cp+as+ck\}$}
is a resolvent for the same inequalities~\cite{CutChase}, where
$m := \lcm\left(a,\frac{ad}{\gcd(ad,c)}\right)-1$.
Note that there is still an existential quantifier $\exists K$ but all variables $k \in K$ are guarded by their respective $R_k$.

Let $\nu$ be a satisfiable assignment for the formula after one weak Cooper elimination step on $C(x)$.
Then we compute a strictest lower bound  $x \geq l_x$ and a strictest upper bound $x \leq u_x$ from $C(x)$ for the variable $x$ under the assignment $\nu$. 
We now argue that there is a value $v_x$ for $x$ such that $x \geq l_x$, $x \leq u_x$, and $d \mid c v_x + s$ are satisfied.
Whenever $l_x \neq -\infty$ and $u_x \neq \infty$, the bounds $x \geq l_x$, $x \leq u_x$ are given by respective constraints of the form $-a x + p \leq 0 \in C(x)$ and $b x - q \leq 0 \in C(x)$ such that $l_x = \lceil \frac{\nu(p)}{a} \rceil$ and $u_x = \lfloor \frac{\nu(q)}{b} \rfloor$.
In this case the extension of $\nu$ with $\nu(x) = \frac{\nu(k+p)}{a}$ satisfies $C(x)$ because the constraint $a \mid k+p \in R_k$ guarantees that $\nu(x) \in \mathbb{Z}$, the constraint $bp-aq+bk \leq 0 \in R_k$ guarantees that $l_x \leq \nu(x) \leq u_x$, and the constraint $ad\mid cp+as+ck \in R_k$ guarantees that $\nu$ satisfies $d \mid c x + s \in C(x)$.
Whenever $l_x = -\infty$ ($u_x = \infty$) we extend $\nu$ by an arbitrary small (large) value for $x$ that satisfies $d \mid c x + s \in C(x)$.
There exist arbitrarily small (large) solutions for $x$ and $d \mid c x + \nu(s)$ because $\gcd(c,d) \mid s$ is satisfied by $\nu$.

The advantage of weak Cooper elimination compared to Cooper elimination is that the output is again one conjunctive problem in contrast to a disjunction of problems.
Our CUTSAT++ performs weak Cooper elimination not in one step but subsequently adds to the states the constraints from the $R_k$ as well as the divisibility constraint $\gcd(c,d) \mid s$ with respect to a strict ordering on the unguarded variables.

The following equivalence, for which we have just outlined the proof, states the correctness of weak Cooper elimination: \newline

\centerline{$\exists x: C(x) \equiv \exists K: \left(\{I \in C(x) : \coeff(I,x) = 0 \} \cup 
                                  \{\gcd(c,d) \mid s\} \cup \underset{k \in K}{\bigcup} R_k \right).$}

The extra  divisibility constraint $\gcd(c,d) \mid s$ in weak Cooper elimination is necessary whenever 
the problem $C(x)$ has no 
constraint of the form $-ax + p \leq 0 \in C(x)$ or $bx - q \leq 0 \in C(x)$. For example,
let $C(x) = \{y - 1 \leq 0, - y + 1 \leq 0,  6 \mid 2 x + y \}$ be a problem and $x$ be the 
unguarded variable we want to eliminate.
As there are no inequalities containing $x$, weak Cooper elimination without the extra divisibility constraint
returns $C'= \{y - 1 \leq 0, - y + 1 \leq 0\}.$
While $C'$ has a satisfiable assignment $\nu(y) = 1$, 
$C(x)$ has not since $2 x + 1$ is never divisible by $2$ or $6$.

%Note that for any $R_k$ introduced by weak Cooper elimination we can also show\newline
%\centerline{$\exists x: (-a x + p \leq 0 \wedge b x - q \leq 0 \wedge d \mid c x + s) \equiv 
%            \exists k: R_k$ .}
%That means satisfiability of the respective $R_k$ guarantees a solution for the triple of constraints 
%it is derived from.

Note that for any $R_k$ introduced by weak Cooper elimination we can also show the following Lemma:\newline
\begin{lemma}
\label{lemma: cooper-wc equivalence}
Let $k$ be a new variable. Let $a,b,c > 0$. Then,\newline
\centerline{$
\begin{array}{c}
(\exists x: \{-ax + p \leq 0, bx - q \leq 0, d \mid c x + s\}) \\ \equiv (\exists k: \{-k \leq 0, k - m \leq 0, bp-aq+bk \leq 0, a \mid k+p, \,ad\mid cp+as+ck\}).
\end{array}$}
\end{lemma}
\begin{proof}
See~\cite{CutChase} pp. 101-102 Lemma 4. 
\end{proof}

That means satisfiability of the respective $R_k$ guarantees a solution for the triple of constraints 
it is derived from.
An analogous Lemma holds for the divisibility constraint $gcd(c,d) \mid s$ introduced by weak Cooper elimination:
\begin{lemma}
\label{lemma: cooper-div equivalence}
$(\exists x: d \mid c x + s) \equiv \gcd(c,d) \mid s.$
\end{lemma}
\begin{proof}
We equivalently rewrite the two divisibility constraints into diophantine equations, viz. $\exists y: d y - c x = s$, and $\exists k: \gcd(c,d) k = s$
for $d \mid c x + s$, and $gcd(c,d) \mid s$, respectively.
We choose $d'$, $c' \in \mathbb{Z}$ such that $d' \cdot \gcd(c,d) = d$ and $c' \cdot \gcd(c,d) = c$.
Assume that $\nu$ is a variable assignment such that
$d \nu(y) - c \nu(x) = \nu(s) \mbox{ and therefore also } d \mid c \nu(x) + \nu(s).$
Thence $\nu(s) = d \nu(y) - c \nu(x) = \gcd(c,d) \cdot  (d' \nu(y) - c' \nu(x)).$
After extending $\nu$ with $\nu(k) = (d' \nu(y) - c' \nu(x))$, $\nu$ satisfies $\gcd(c,d) k = s$.

Assume, that $\nu$ is a variable assignment such that $\gcd(c,d) \nu(k) = \nu(s)$ holds and therefore also $\gcd(c,d) \mid \nu(s)$.
By Be\'zout's Lemma there exist $a', b' \in \mathbb{Z}$, such that $a' d - b' c = \gcd(c,d)$.
Thence $a' d  \nu(k) - b' c \nu(k) = (a' d - b' c)\nu(k) = \gcd(c,d)\nu(k) = \nu(s).$
After extending $\nu$ with $\nu(y) = a' \nu(k)$ and $\nu(x) = b' \nu(k)$ the assignment $\nu$ satisfies $d y - c x = s$.
\end{proof}

That means satisfiability of $\gcd(c,d) \mid s$ guarantees a solution for the divisibility constraint $d \mid c x + s$.
The rule Resolve-Cooper~(Fig.~\ref{fig: strong-conflict-rules}) in our CUTSAT++ exploits these properties by generating the $R_k$ and
constraint $\gcd(c,d) \mid s$ in the form of strong resolvents in a lazy way. 
Furthermore, it is not necessary for the divisibility constraints to be a priori reduced to one,
as done for weak Cooper elimination. Instead, the rules Solve-Div-Left and Solve-Div-Right~(Fig.~\ref{fig: strong-conflict-rules}) perform lazy reduction.

The \emph{solution set} for variable $x$, assignment $\nu$, and problem $C(x)$, is the set of values $S \subseteq \mathbb{Z}$ such that 
$v \in S$ if $C(v)$ is satisfied by $\nu$.
The solution set $S_d$ of a divisibility constraint $d \mid c x + s$, variable $x$, and assignment $\nu$ is either empty or unbounded from above and below.

\begin{lemma}
\label{lemma: cooper-div solution set}
Let $\nu$ be an assignment for all variables except $x$. Let $S_d$ be the solution set for variable $x$, assignment $\nu$, and constraint $d \mid c x + s$.
Then, \newline
\centerline{$S_d = \emptyset \mbox{ or } S_d = \{v_0 + e v' : e \in \mathbb{Z}\} \mbox{ for some } v_0, v' \in \mathbb{Z}.$}
\end{lemma}
\begin{proof}
In case $S_d \neq \emptyset$, there exists a value $v_0 \in S_d$ such that $d  \mid c v_0 + \nu(s)$.
We first prove that there exists a $v' \in \mathbb{Z}$ such that 
$d  \mid c (v_0 + e v') + \nu(s)$
for all $e \in \mathbb{Z}$ and therefore $v_0 + e v' \in S_d$.
We choose $v'$, $e' \in \mathbb{Z}$ such that $c = e' \gcd(c,d)$ and $d = v' \gcd(c,d)$.
Then we deduce for any $e \in \mathbb{Z}$:\newline
\centerline{$
\begin{array}{c}
d \mid c (v_0 + e v') + \nu(s) \equiv d \mid c v_0 + \nu(s) + c e v' \equiv \\
d \mid c v_0 + \nu(s) + d e' v' \equiv d \mid c v_0 + \nu(s)
\end{array}
$}

It remains to show that for every $v_k \in S_d$ there exists an $e \in \mathbb{Z}$ such that $v_0 + e v' = v_k$.
As $S_d$ is the solution set we know that $d  \mid c v_0 + \nu(s)$ and $d  \mid c v_k + \nu(s)$ are true.
Thence $d  \mid c (v_0 - v_k) \equiv d  \mid c v_0 + \nu(s) - (c v_k + \nu(s)).$
As $d = v' \gcd(c,d)$, the term $c (v_0 - v_k)$ is only divisible by $d$ if $v_0 - v_k$ is divisible by $v'$.
Therefore, $\exists e \in \mathbb{Z}: v_0 - v_k = e v'$.
\end{proof}

This property allows us to choose an arbitrary small or large solution for $x$ to satisfies $d \mid c x + \nu(s)$ in the correctness proof of weak Cooper Elimination.
As mentioned in the outline of the proof, the ability to choose arbitrary small and large solutions for $x$ is necessary when $C(x)$ contains no constraints of the form $-ax + p \leq 0$ or $bx - q \leq 0$.

\begin{theorem}\ \newline
\label{lemma: Weak-Elim equivalence}
\centerline{$\exists x: C(x) \equiv \exists K: \left(\{I \in C(x) : \coeff(I,x) = 0 \} \cup 
                                  \{\gcd(c,d) \mid s\} \cup \underset{k \in K}{\bigcup} R_k \right)$}
\end{theorem}
\begin{proof}
First, we partition the problem $C(x)$ as follows:\newline
\centerline{$
  \begin{array}{l l l l}
    C_l &= \{-ax+p \leq 0 \in C(x) : a > 0 \}, &C_u &= \{bx-q \leq 0 \in C(x) : b > 0 \},\\
    I_d &= d \mid c x + s \in C(x), &C_r &= \{I \in C(x) : \coeff(I,x) = 0 \}.
  \end{array}
$}
%To prove the equivalence $(\exists x: C') \equiv (\exists K: \WCElim(x,C))$ we prove $(\exists x: C') \rightarrow (\exists K: \WCElim(x,C))$.
By Lemma~\ref{lemma: cooper-wc equivalence}, it holds for all $-ax + p \leq 0, bx - q \leq 0 \in C(x)$ with $a,b > 0$ that:\newline
\centerline{$
\begin{array}{c}
(\exists x: C(x)) \rightarrow (\exists x: \{-ax + p \leq 0, bx - q \leq 0, d \mid c x + s\}) \\ \rightarrow (\exists k: \underbrace{\{-k \leq 0, k - m \leq 0, bp-aq+bk \leq 0, a \mid k+p, \,ad\mid cp+as+ck\}}_{R_k}).
\end{array}
$}
By Lemma~\ref{lemma: cooper-div equivalence}, it holds that: $(\exists x: C(x)) \rightarrow (\exists x: d \mid c x + s) \rightarrow \gcd(c,d) \mid s.$
As $C_r \subseteq C(x)$ it also holds that: $(\exists x: C(x)) \rightarrow C_r.$
As all new variables $k \in K$ appear only in one resolvent $R_k$, the above implications prove\newline
\centerline{$\exists x: C(x) \rightarrow \exists K: \left(\{I \in C(x) : \coeff(I,x) = 0 \} \cup 
                                  \{\gcd(c,d) \mid s\} \cup \underset{k \in K}{\bigcup} R_k \right).$}

Assume, vice versa, that $\nu$ is a satisfiable assignment for the formula after one step of weak Cooper elimination. Then it is easy to deduce the following facts:
\begin{itemize}
\item Let $S_l$ be the solution set for $x$, $\nu$, and $I_l = -ax + p \leq 0 \in C(x)$ with $a > 0$. Then $S_l = \left\{\lceil\frac{\nu(p)}{a}\rceil, \lceil\frac{\nu(p)}{a}\rceil + 1,\ldots\right\}$.
\item Let $S_u$ be the solution set for $x$, $\nu$, and $I_u = bx - q \leq 0 \in C(x)$ with $b > 0$. Then $S_u = \left\{\ldots,\lfloor\frac{\nu(q)}{b}\rfloor-1,\lfloor\frac{\nu(q)}{b}\rfloor\right\}$.
\item Let $S_I$ be the solution set for $x$, $\nu$, and $C_l \cup C_u$. Then $S_I = \bigcap_{I_l \in C_l} S_l \cap \bigcap_{I_u \in C_u} S_u$.
\item Let the set $S_I$ be bounded from below, i.e., $S_I = \{l,l+1,\ldots\}$ or $S_I = \{l,\ldots,u\}$. Then $l = \max_{I \in C_l}\left\{\lceil\frac{\nu(p')}{a'}\rceil : I = -a' x + p' \leq 0 \right\}$.
\item Let the set $S_I$ be bounded from above, i.e., $S_I = \{\ldots,u-1,u\}$ or $S_I = \{l,\ldots,u\}$. Then $u = \min_{I \in C_u}\left\{\lfloor\frac{\nu(q)}{b}\rfloor : I = b' x - q' \leq 0 \right\}$.
\item By Lemma~\ref{lemma: cooper-div equivalence}, $d \mid c x + \nu(s)$ is satisfiable because $\gcd(c,d) \mid s$ is contained in the result formula of weak Cooper elimination.
      By Lemma~\ref{lemma: cooper-div solution set}, the set of solutions for $x$, $\nu$, and $d \mid c x + s$ has the form $S_d = \{ v_0 + v' \mu : v' \in \mathbb{Z} \}$.
\item The solution set $S$ for $x$, $\nu$, and $C'$ is $S = S_d \cap S_I$.
\end{itemize}
Next, we do a case distinction on the structure of $C(x)$:
\begin{itemize}
\item Let $C_l = \emptyset$, then $S_I$ is unbounded from below.
      We choose a small enough $v \in S_d$, i.e., small enough $v' \in \mathbb{Z}$ such that $v = v_0 + v' \mu$.
      Then the assignment $x \mapsto v$ and $y \mapsto \nu(y)$ (if $y \neq x$) satisfies $C'$.
\item Let $C_u = \emptyset$, then $S_I$ is unbounded from above.
      We choose a large enough $v \in S_d$, i.e., large enough $v' \in \mathbb{Z}$ such that $v = v_0 + v' \mu$.
      Then the assignment $x \mapsto v$ and $y \mapsto \nu(y)$ for all $y \neq x$ satisfies $C'$.
\item Let $|C_l|, |C_u| > 0$. We select $I_l = -ax + p \leq 0$ such that\newline
\centerline{$
        \lceil\frac{\nu(p)}{a}\rceil = \max_{I \in C_l}\left\{\lceil\frac{\nu(p')}{a'}\rceil : I = -a' x + p' \leq 0\right\}
      $}
      and $I_u = bx - q \leq 0$ such that\newline
      \centerline{$
        \lfloor\frac{\nu(q)}{b}\rfloor = \min_{I \in C_u}\left\{\lfloor\frac{\nu(q')}{b'}\rfloor : I = b' x - q' \leq 0\right\}.
      $}
      The resolvent for the two constraints $I_l$ and $I_u$ is\newline
      \centerline{$
        R_k = \{-k \leq 0, k - m \leq 0 , bp-aq+bk \leq 0, a \mid k+p, \,ad\mid cp+as+ck \}.
      $}
      We will now show that $\frac{\nu(p + k)}{a}$ is in the set of solutions $S$ of $C(x)$.
      All of the remaining deductions stem from the evaluation of the resolvent under $\nu$.
      Since $a \mid \nu(p + k)$, $\frac{\nu(p + k)}{a} \in \mathbb{Z}$.
      Furthermore, since $\frac{\nu(p + k)}{a} \in \mathbb{Z}$ and $\nu(bp-aq+bk) \leq 0$, $\frac{\nu(p + k)}{a} \in S_I = \left\{ \lceil \frac{\nu(p)}{a} \rceil, \ldots, \lfloor \frac{\nu(q)}{b} \rfloor \right\}$.
      Finally, since \centerline{$
      a d \mid \nu(c p + a s + c k) = a d \mid a c \nu(x) + a \nu(s) = d \mid c \nu(x) + \nu(s),$}
       $\frac{\nu(p + k)}{a} \in S_d$.
      We choose the assignment $\nu'$ with $x \mapsto \frac{\nu(p + k)}{a}$ and $y \mapsto \nu(y)$ for all $y \neq x$.
      Hence, $\nu'$ satisfies $C'$.
\end{itemize}
\end{proof}

\begin{figure}[t]
\begin{minipage}{\linewidth}
\begin{algorithm}[H]
    \algoIfrule
    \caption{$\CombineDivs(x,C'(x))$}
    \label{ALGO: Phase I}
    \Input{The variable $x$ and a set of LIA constraints $C'(x)$}
    \Output{A set of LIA constraints $C(x)$ such that $C(x) \equiv C'(x)$ and there exists one divisibility constraint $d \mid c x + s \in C(x)$ such that $c > 0$}
    $C_d := \{d \mid c x + s \in C'(x) : c > 0 \}$\\
    $C(x)  := C'(x) \setminus C_d$ \;
    \uIfrule{$(C_d = \emptyset)$}{
      \Return $C(x) \cup \{ 1 \mid x\}$ \label{step C_d empty} \;}
    \While{$(|C_d| > 1)$}{ \label{step div-start}
      Select  $d_1 \mid a_1 x + p_1, d_2 \mid a_2 x + p_2 \in C_d$ \;
      $C_d := C_d \setminus \{d_1 \mid a_1 x + p_1, d_2 \mid a_2 x + p_2 \}$\;
      $d = \gcd(a_1 d_2, a_2 d_1)$ \;
      Choose $c_1$ and $c_2$ such that $c_1 a_1 d_2 + c_2 a_2 d_1 = d$\;
      $C_d := C_d \cup \{d_1 d_2 \mid d x + c_1 d_2 p_1 + c_2 d_1 p_2 \}$ \label{step div-solve}\;
	  $C(x) := C(x) \cup \{d \mid -a_1 p_2 + a_2 p_1\}$\; 
    } \label{step div-end}
    \Return $C(x) \cup C_d$ \label{step phaseI res}\;
\end{algorithm}
\end{minipage}
\caption{An algorithm that combines constraints $C_d = \{d \mid c x + s \in C'(x) : c > 0 \}$ until only one divisibility constraint in $x$ remains}
\label{fig: ALGO: Phase I}
\end{figure}

We stated that weak Cooper elimination can only be applied to those problems where $C(x)$ contains one divisibility constraint $d \mid a x + p$ in $x$.
To expand weak Cooper elimination to any set of constraints $C'(x)$ we briefly explained how to exhaustively apply $\divsolve$ to eliminate all but one constraint $d \mid a x + p$ in $x$.
The algorithm $\CombineDivs(x,C)$ (Fig.~\ref{fig: ALGO: Phase I}) is a more detailed version of this procedure.

\begin{lemma}
\label{lemma: Weak-Elim-I equivalence}
Let $C'(x)$ be a set of LIA constraints. Let $C(x)$ be the output of $\WCElimI(x,C'(x))$. Then $C(x) \equiv C'(x).$
\end{lemma}
\begin{proof}
Follows directly from the proof of equivalence~\cite{CutChase} of the $\divsolve$ transformation.
\end{proof}

Since the output $C(x)$ of $\WCElimI(x,C'(x))$ is equivalent to $C'(x)$ and fulfils the conditions of weak Cooper elimination, we conclude the following equivalence for the output of weak Cooper elimination applied to $C(x)$:

\centerline{$\exists x: C'(x) \equiv \exists K: \left(\{I \in C(x) : \coeff(I,x) = 0 \} \cup 
                                  \{\gcd(c,d) \mid s\} \cup \underset{k \in K}{\bigcup} R_k \right)$}

  \section{Strong Conflict Resolution Revisited}
\label{SE: Strong Conflict Resolution}
\label{SSE: Strong Conflict Resolution}

Weak Cooper elimination is capable of exploring all unguarded variables to eventually create a problem where feasibility only depends on guarded variables.
It is simulated in a lazy manner through an additional set of CUTSAT++ rules (Fig.~\ref{fig: strong-conflict-rules}).
Instead of eliminating all unguarded variables before the application of CUTSAT++, 
the rules perform the same intermediate steps as weak Cooper elimination, 
viz. the combination of divisibility constraints via $\divsolve$ and the construction of resolvents, to resolve and block conflicts in unguarded constraints.
As a result, CUTSAT++ can avoid some of the intermediate steps of weak Cooper elimination. 
Furthermore, CUTSAT++ is not required to apply the intermediate steps of weak Cooper elimination one variable at a time.
The lazy approach of CUTSAT++ does not eliminate unguarded variables. 
In the worst case CUTSAT++ has to perform all of weak Cooper elimination's intermediate steps.
Then a strategy (Def.~\ref{definition:strictly-two-layered}) guarantees 
that CUTSAT++ recognizes that all unguarded conflict constraints have been blocked by guarded constraints.

The eventual result is the complete algorithm CUTSAT++, which is a combination of the rules Resolve-Cooper, Solve-Div-Left, Solve-Div-Right~(Fig.~\ref{fig: strong-conflict-rules}), a strictly-two-layered strategy~(Def.~\ref{definition:strictly-two-layered}), and the CUTSAT rules: Propagate, Propagate-Div, Decide, Conflict, Conflict-Div, Sat, Unsat-Div, Forget, Slack-Intro, Resolve, Skip-Decision, Backjump, Unsat, and Learn~\cite{CutChase}.

%In WC remove this paragraph
The advantage of the lazy approach is that CUTSAT++ might find a satisfiable assignment 
or detect unsatisfiability without 
encountering and resolving a large number of unguarded conflicts.
This means the number of divisibility constraint combinations and introduced resolvents 
might be much smaller in the lazy approach of CUTSAT++ than during the elimination with weak Cooper elimination.

In order to simulate weak Cooper elimination, CUTSAT++ uses a total order $\prec$ over all variables 
such that $y \prec x$ for all guarded variables $y$ and unguarded variables $x$~\cite{CutChase}.
While the order needs to be fixed for all unguarded variables, the ordering
among the guarded variables can be dynamically changed.
In relation to weak Cooper elimination, the order $\prec$ describes the elimination order for the unguarded variables, 
viz. $x_i \prec x_j$ if $x_j$ is eliminated before $x_i$.
A variable $x$ is called \emph{maximal} in a constraint $I$ if $x$ is contained in $I$ 
and all other variables in $I$ are smaller with respect to $\prec$.
The maximal variable in $I$ is also called its \emph{top variable} $(x = \topop(I))$.

\begin{definition} \label{definition:conflicting cores}
Let $S = \SearchSt{M}{C}$ be a state,  $C' \subseteq C$, $x$ the top variable in $C'$ and 
let all other variables in $C'$ be fixed. The pair $(x,C')$ is a \emph{conflicting core} if it is of one of the following three forms\newline
(1)~$C' = \{-ax+p \leq 0, bx-q \leq 0\}$, and the lower bound from $-ax+p \leq 0$ contradicts the upper 
bound from $bx-q \leq 0$, i.e., $\bound(-ax+p \leq 0,x,M) > \bound(bx-q \leq 0,x,M)$\cite{CutChase};
in this case $(x,C')$ is called an \emph{interval conflicting core} and its strong resolvent is
$(\{-k \leq 0, k - a + 1 \leq 0 \}, \{ bp-aq+bk \leq 0, a \mid k+p \})$~\cite{CutChase}\newline
(2)~$C' = \{-ax+p \leq 0, bx-q \leq 0, d \mid c x + s\}$, and $b_l = \bound(-ax+p \leq 0,x,M)$,
$b_u = \bound(bx-q \leq 0,x,M)$, $b_l \leq b_u$ and for all $b_d \in [b_l, b_u]$ we have $d \nmid c b_d + \lowerop(s,M)$\cite{CutChase}; in this case $(x,C')$ is called a
\emph{divisibility conflicting core} and its strong resolvent is
$(\{-k \leq 0, k - m \leq 0 \}, \{ bp-aq+bk \leq 0, a \mid k+p, ad\mid cp+as+ck \})$~\cite{CutChase}\newline
(3)~$C' = \{d \mid c x + s\}$, and for all $b_d \in \mathbb{Z}$ we have $d \nmid c b_d + \lowerop(s,M)$;
in this case $(x,C')$ is called a \emph{diophantine conflicting core} and its strong resolvent is
$( \emptyset ,\, \{ gcd(c,d) \mid s \})$.\newline In the first two cases $k$ is a fresh variable and $m = \lcm\left(a,\frac{ad}{\gcd(ad,c)}\right)-1$.
\end{definition}

We refer to the respective strong resolvents for a conflicting core $(x,C')$ by
the function $\cooper(x,C')$ which returns a pair $(R_k,R_c)$ as defined above.
Note that the newly introduced variable $k$ is guarded by the constraints in $R_k$.
If there is a conflicting core $(x,C')$ in state $S$, then $x$ is called a \emph{conflicting variable}.
%The conflicting variable $x$ is the \emph{minimal conflicting variable} if there is no conflicting variable $y$ that is smaller, i.e., $y \prec x$.
A \emph{potential conflicting core} is a pair $(x,C')$ if there exists a state $S$ where $(x,C')$ 
is a conflicting core.

% Suppose we are in a state $S$ in which the \emph{key-invariant} holds and we apply 
% either Div-Solve-Left or Div-Solve-Right.
% If $I'_2$ is a conflict we have to apply Div-Solve-Right and 
% backjump to a state where the \emph{key-invariant} holds.
% If $I'_1$ is a conflict and $x$ is fixed by a decision, 
% then one of the constraints $I_1$ and $I_2$ must be a conflict.
% The last case is impossible because a conflict 
% in $I_1$ or $I_2$ violates the \emph{key-invariant} in state $S$.

Next we define a semantic generalization of strong resolvents. Since the strong resolvents
generated out of conflicting cores will be further processed by CUTSAT++, we must guarantee
that any set of constraints implying the feasibility of the conflicting core constraints
prevents a second application of Resolve-Cooper to the same conflicting core. All strong resolvents
of Definition~\ref{definition:conflicting cores} are also strong resolvents in the 
sense of the below definition (see also end of Section~\ref{SE: Weak-Cooper-Elimination}).

\begin{definition}
  \label{definition:resolvent}
  A set of constraints $R$ is a \emph{strong resolvent} for the pair $(x,C')$ if it holds 
 that $R \rightarrow \exists x : C' \mbox{ and for all } J \in R: \topop(J) \prec x.$
\end{definition}

\begin{lemma}
  \label{lemma:strong resolvent}
  Let $C' \subseteq C$. Let $\cooper(x,C')=(R_k, R_c)$. Let $R = R_k \cup R_c$.
  Then $\exists k: C \cup R \equiv C$. Furthermore, $R$ is a strong resolvent for $(x,C')$.
\end{lemma}
\begin{proof}
  Follows directly from the Lemmas~\ref{lemma: cooper-wc equivalence} and~\ref{lemma: cooper-div equivalence}.
  The interval conflicting core is the only new case.
  However, $\cooper(x,\{ -ax+p \leq 0, bx-q \leq 0 \})$ is equivalent to $\cooper(x,\{-ax+p \leq 0, bx-q \leq 0, 1 \mid x\})$.
  By Lemmas~\ref{lemma: cooper-wc equivalence} and~\ref{lemma: cooper-div equivalence}, $R \rightarrow \exists x: C'$.
  Finally, since $k$ is the minimal element of $\prec$ and all other variables in $R$ appear in $C'$, where $x$ is maximal, it holds that $J \in R: \topop(J) \prec x$.
\end{proof}

The rule Resolve-Cooper (Fig.~\ref{fig: strong-conflict-rules})
requires that the conflicting variable $x$ of the conflicting core $(x,C')$ is the 
top-variable in the constraints of $C'$.
This simulates a setting where all variables $y$ with $x \prec y$ are already eliminated.
We restrict Resolve-Cooper to unguarded constraints, because weak Cooper elimination modifies only unguarded constraints.

\begin{lemma}
  \label{lemma:resolvent conflict}
  Let $S = \SearchSt{M}{C}$ be a CUTSAT++ state.
  Let $C' \subseteq C$ and $x$ be an unguarded variable.
  Let $R$, $R \subseteq C$, be a \emph{strong resolvent} for $(x,C')$.
  Then Resolve-Cooper is not applicable to $(x,C')$.\qed
\end{lemma}
\begin{proof}
  Assume for contradiction that $D = (x,C')$ is a conflicting core, $R \in C$ is a strong resolvent for $D$ in state $S$ and Resolve-Cooper is applicable to $D$ in state $S$.
  Resolve-Cooper requires that all variables $y \prec x$ are fixed (Fig.~\ref{fig: strong-conflict-rules}). This holds especially for all variables in $R$ (Def.~\ref{definition:resolvent}).
  Due to the restriction that every conflict $J \in C$ has $\topop(J) \not \prec \topop(I)$ in Resolve-Cooper, there is no conflict in $R$.
  Furthermore, since all variables $y \prec x$ are fixed, $R$ is satisfied by the partial assignment defined by $M$.
  By Def.~\ref{definition:conflicting cores}, all conflicting cores have no satisfiable solution for $x$ under partial model $M$.
  However, by Def.~\ref{definition:resolvent}, $R$ satisfiable implies that there exists an $x$ such that $C'$ is satisfiable under $M$.
  This contradicts the assumption that $(x,C')$ is a conflicting core!
\end{proof}

For the resolvent $R$ to block Resolve-Cooper from being applied to the conflicting 
core $(x,C')$, CUTSAT++ has to detect all conflicts in $R$.
Detecting all conflicts in $R$ is only possible if CUTSAT++ fixes all variables 
$y$ with $y \prec x$ and if Resolve-Cooper is only applicable if there exists no 
conflict $I$ with $\topop(I) \prec x$.
Therefore, the remaining restrictions of Resolve-Cooper justify the above Lemma.
If we add strong resolvents again and again, then CUTSAT++ will reach a state after 
which every encounter of a conflicting core guarantees a conflict in a guarded constraint.
From this point forward CUTSAT++ won't apply Resolve-Cooper anymore. 
The remaining guarded conflicts are resolved with the rules Conflict and Conflict-Div~\cite{CutChase}.

The rules Solve-Div-Left and Solve-Div-Right (Fig.~\ref{fig: strong-conflict-rules}) 
combine divisibility constraints as it is done a priori to weak Cooper elimination.
In these rules we restrict the application of $\divsolve(x,I_1,I_2)$ to constraints where 
$x$ is the top variable and where all variables $y$ in $I_1$ and $I_2$, with $y \neq x$, are fixed.
The ordering restriction simulates the order of elimination, 
i.e., we apply $\divsolve(x,I_1,I_2)$ in a setting where all variables 
$y$ with $x \prec y$ appear to be eliminated in $I_1$ and $I_2$.
Otherwise, divergence would be possible (see example~\ref{example: Solve-Div top-var restriction}).
Requiring smaller variables to be fixed prevents the accidental generation of a conflict
for an unguarded variable $x_i$ by $\divsolve(x,I_1,I_2)$.

Thanks to an eager top-level propagating strategy, defined below, 
any unguarded conflict in CUTSAT++ is either resolved with Solve-Div-Right 
(Fig.~\ref{fig: strong-conflict-rules}) or 
CUTSAT++ constructs a conflicting core that is resolved with Resolve-Cooper. 
Both cases may require multiple applications of the 
Solve-Div-Left rule (Fig.~\ref{fig: strong-conflict-rules}).
We define the following further restrictions on the CUTSAT++ rules that will eventually generate
the above described behavior.

\begin{definition}
  \label{def:eager top}
  Let ${\bowtie} \in \{\leq, {\geq}\}$.
  We call a strategy for CUTSAT++ \emph{eager top-level propagating} if we restrict propagations and decisions for every state $\SearchSt{M}{C}$ in the following way:
  \begin{enumerate}
  \item Let $x$ be an unguarded variable. Then we only allow to propagate bounds $\propBnd{x}{\bowtie}{\bound(I,x,M)}{}$ if $x$ is the top variable in $I$.
    Furthermore, if $I$ is a divisibility constraint $d \mid a x + p$, then we only propagate  $d \mid a x + p$ if:
    \begin{enumerate}
    \item Either $\lowerop(x,M) \neq -\infty$ and $\upper(x,M) \neq \infty$ 
    \item Or if $\gcd(a,d) \mid \lowerop(p,M)$, and $d \mid a x + p$ is the only divisibility constraint in $C$ with $x$ as top variable.
    \end{enumerate}
  \item Let $x$ be an unguarded variable. Then we only allow decisions $\gamma = \decBnd{x}{\bowtie}{b}$ if:
    \begin{enumerate}
    \item For every constraint $I \in C$ with $x = \topop(I)$ all occurring variables $y \neq x$ are fixed
    \item There exists no $I \in C$ where $x = \topop(I)$ and $I$ is a conflict in $\concseq{ M,\gamma }$
    \item Either $\lowerop(x,M) \neq -\infty$ and $\upper(x,M) \neq \infty$ or there exists at most one divisibility constraint in $C$ with $x$ as top variable.
    \end{enumerate}
  \end{enumerate}
\end{definition}

An eager top-level propagating strategy has two advantages. 
First, the strategy dictates an order of influence over the variables, i.e., a bound for unguarded variable $x$ is only influenced by previously propagated bounds for variable $y$ with $y \prec x$.
Furthermore, the strategy makes only decisions for unguarded variable $x$ when all constraints with $x = \topop(I)$ are fixed and satisfied by the decision.
This means any conflict $I \in C$ with $x = \topop(I)$ is impossible as long as the decision for $x$ remains on the bound sequence.
For the same purpose, i.e., avoiding conflicts $I$ where $x = \topop(I)$ is fixed by a decision, CUTSAT++ backjumps in the rules Resolve-Cooper and Solve-Div-Right to a state where this is not the case.
To avoid frozen states resulting from the eager top-level propagating strategy, the slack variable $x_S$ has to be the smallest unguarded variable in $\prec$. 
Otherwise, the constraints $x - x_S \leq 0$, $-x - x_S \leq 0$ introduced by Slack-Intro cannot be used to propagate bounds for variable $x$, and $x$ would remain stuck.

\begin{figure}[t]
\begin{align*}
\begin{CSRules}
   \CSRuleName{Solve-Div-Left}
   \CSRuleME{\SearchSt{M}{C}}{\SearchSt{M}{C'}}
    {\mbox{divisibility constraints } I_1,I_2 \in C, \\ x \mbox{ is top in } I_1 \mbox{ and } I_2, \\ x \mbox{ is unguarded},\\ \mbox{all other vars. in } I_1, I_2 \mbox{ are fixed}, \\
     (I'_1,I'_2) = \divsolve(x, I_1,I_2),\\ C' = C \setminus \lbrace I_1, I_2 \rbrace \cup \lbrace I'_1, I'_2 \rbrace, \\ I'_2 \mbox{ is not a conflict}}
   \CSRuleName{Solve-Div-Right}
   \CSRuleME{\SearchSt{M}{C}}{\SearchSt{M'}{C'}}
   {\mbox{divisibility constraints } I_1,I_2 \in C, \\ x \mbox{ is top in } I_1 \mbox{ and } I_2, \\ x \mbox{ is unguarded},\\ \mbox{all other vars. in } I_1, I_2 \mbox{ are fixed}, \\
    (I'_1,I'_2) = \divsolve(x, I_1,I_2),\\ C' = C \setminus \lbrace I_1, I_2 \rbrace \cup \lbrace I'_1, I'_2 \rbrace, \\ I'_2 \mbox{ is a conflict} \\
    y  = \topop(I'_2) \\ M' = \prefix(M,y)}
   \CSRuleName{Resolve-Cooper}
   \CSRuleME{\SearchSt{M}{C}}{\SearchSt{M'}{C \cup R_k \cup R_c}}
   {(x,C') \mbox{ is a conflicting core}, \\ x \mbox{ is unguarded},\\ \mbox{ all } z \prec x \mbox{ are fixed and } C' \subseteq C, \\
    \mbox{if } J \in C \mbox{ is a conflict, then } \topop(J) \not \prec x, \\ \cooper(x,C') = (R_k,R_c),\\ y = \min_{I \in R_c}\{\topop(I)\},\\ M' = \prefix(M,y)}
\end{CSRules}
\end{align*}
In the above rules, $M' = \prefix(M,y)$ defines the largest prefix of $M$ that contains only decided bounds for variables $x$ with $x \prec y$.
\caption{Our strong conflict resolution rules}
\label{fig: strong-conflict-rules}
\end{figure}

\begin{definition}
\label{def: reasonable strategy}
A strategy is \emph{reasonable} if Propagate applied to constraints of the form $\pm x - b \leq 0$ has the highest priority over all rules and the Forget Rule is applied only finitely often \cite{CutChase}.
\end{definition}

\begin{definition}
  \label{definition:strictly-two-layered}
  A strategy is \emph{strictly-two-layered} %for an initial state $\SearchSt{\emptyseq}{C_0}$
   if:\newline (1) it is reasonable, (2) it is eager top-level propagating (Def.~\ref{def:eager top}),
  (3) the Forget, Conflict, Conflict-Div rules only apply to guarded constraints,
  (4) Forget cannot be applied to a divisibility constraint or a constraint contained in a strong resolvent, and
  (5) only guarded constraints are used to propagate guarded variables.
\end{definition}

A \emph{strictly-two-layered} strategy is the final restriction to CUTSAT++.
With the condition~\ref{definition:strictly-two-layered}-(3) it partitions conflict resolution into 
two layers:
Every unguarded conflict is handled with the rules Resolve-Cooper, Solve-Div-Left, and Solve-Div-Right (Fig.~\ref{fig: strong-conflict-rules}), 
every guarded conflict with the rules Conflict(-Div).
The conditions~\ref{definition:strictly-two-layered}-(1) and~\ref{definition:strictly-two-layered}-(5) 
make the guarded variables independent from the unguarded variables.
The conditions~\ref{definition:strictly-two-layered}-(2) and~\ref{definition:strictly-two-layered}-(4)  
give a guarantee that the rules Resolve-Cooper, Solve-Div-Left, and Solve-Div-Right are applied at most finitely often.
We assume for the remainder of the paper that all runs of CUTSAT++ follow a strictly-two-layered strategy.

  \section{Termination and Completeness}
\label{SE: Termination and Completeness}

The CUTSAT++ rules are Propagate, Propagate-Div, Decide, Conflict, Conflict-Div, Sat, Unsat-Div, 
Forget, Slack-Intro, Resolve, Skip-Decision, Backjump, Unsat, and Learn from~\cite{CutChase}, 
as well as Resolve-Cooper, Solve-Div-Left, 
and Solve-Div-Right~(Fig.~\ref{fig: strong-conflict-rules}).
Before we prove termination and completeness for CUTSAT++, we have to prove another property over strong resolvents.
We have proven in Section \ref{SE: Strong Conflict Resolution} that Resolve-Cooper applied to conflicting core $(x,C')$ adds a strong resolvent $R$, which blocks another application of Resolve-Cooper to $(x,C')$.
However, CUTSAT++ is able to remove constraints from $R$ with the rules Solve-Div-Left and Solve-Div-Right.
This removes the original conflicting core $R$ from our state.
Nonetheless, CUTSAT++ is still unable to apply Resolve-Cooper to conflicting core $(x,C')$ because the rules Solve-Div-Left and Solve-Div-Right guarantee that a new strong resolvent $R'$ for conflicting core $(x,C')$ is introduced:

\begin{lemma}
  \label{lemma:replacing resolvents}
  Let $S = \SearchSt{M}{C}$ be a state reachable by CUTSAT++ from the initial state $\SearchSt{\emptyseq}{C_0}$ and let
  $S' = \SearchSt{M'}{C'}$ be a state reachable by CUTSAT++ from $S$.
  Let $C$ contain a strong resolvent $R$ for $(x, C'')$. Then $C'$ contains also a strong resolvent $R'$ for $(x, C'')$.
\end{lemma}
\begin{proof}
  Assume for a contradiction that $S$ contains a strong resolvent $R$ for $(x, C'')$ and $S'$ contains no strong resolvent $R' \in C'$ for $(x, C'')$.
  W.l.o.g. we assume that $S'$ is the first state after $S$ where $R \nsubseteq C'$.
  By Def.~\ref{definition:strictly-two-layered}.(4), CUTSAT++ with a strictly-two-layered strategy cannot remove constraints from
  a strong resolvent $R$ except with the rules Solve-Div-Right and Solve-Div-Left.
  Through the equivalence proven for $\divsolve(x,I_1,I_2)$ in \cite{CutChase}, we know that there exist $I'_1, I'_2 \in C'$ such that $\{I'_1, I'_2\} \equiv \{I_1, I_2\}$ and
  $R \subseteq C' \setminus \{I'_1, I'_2\} \cup \{I_1, I_2\}$.
  Thus $R' = R \setminus \{I_1, I_2\} \cup \{I'_1, I'_2\}$ is a strong resolvent of $(x, C'')$ such that \newline
  \centerline{$ R' \rightarrow R \rightarrow \exists x: C''$.}
  Furthermore, $R'$ is a subset of $C'$, which contradicts our initial assumption!
\end{proof} 

Together with Lemma \ref{lemma:resolvent conflict} this property implies that Resolve-Cooper is applied at most once to every conflicting core encountered by CUTSAT++.
This is essential for our termination proof.

\subsection{Proof for Termination}
\label{SSE: Termination of the Strong-Conflict-Rules}

For the termination proof of CUTSAT++, we consider a (possibly infinite) 
sequence of rule applications 
$\SearchSt{\emptyseq}{C_0} = S_0 \CSArrow S_1 \CSArrow \ldots$ on a problem $C_0$,
following the strictly-two-layered strategy.

First, this sequence reaches a state 
$S_s$ ($s \in \mathbb{N}^+_0$) after a finite derivation of rule 
applications $S_0 \CSArrow \ldots \CSArrow S_s$
such that there is no further application of the rules Slack-Intro 
and Forget after state $S_s$: 

\begin{lemma}
  \label{lemma:prepphase}
  Let $\SearchSt{\emptyseq}{C_0} = S_0 \CSArrow S_1 \CSArrow \ldots$ be a sequence of rule applications applied to a problem $C_0$,
following the strictly-two-layered strategy.
  Then the sequence reaches a state $S_s$ ($s \in \mathbb{N}^+_0$) after at most finitely many rule 
applications $S_0 \CSArrow \ldots \CSArrow S_s$
such that there is no further application of the rules Slack-Intro 
and Forget after state $S_s$. 
\end{lemma}
\begin{proof}
Such a state $S_s$ exists for two reasons: 
Firstly, the strictly-propagating-strategy employed by CUTSAT++ is also reasonable. 
The reasonable strategy explicitly forbids infinite applications of the rule Forget. 
Secondly, the Slack-Intro rule is applicable only to stuck variables and only once 
to each stuck variable. 
Only the initial set of variables can be stuck because all variables $x$ introduced 
during the considered derivation are introduced with at least one constraint $x - b \leq 0$ that allows at least 
one propagation for the variable.
Therefore, the rules Slack-Intro and Forget are at most finitely often applicable.
\end{proof}

Next, the sequence reaches a state $S_w$ ($w \geq s$) after 
a finite derivation of rule applications $S_s \CSArrow \ldots \CSArrow S_w$ such that there is 
no further application of the rules Resolve-Cooper, Solve-Div-Left or Solve-Div-Right 
after state $S_w$:
The rules Resolve-Cooper, Solve-Div-Left, Solve-Div-Right, and Slack-Intro are applicable only to unguarded constraints. 
Through the strictly-two-layered strategy they are also the only rules producing unguarded constraints. 
Therefore, they form a closed loop with respect to 
unguarded constraints, which we use in our termination proof.
We have shown in the previous paragraph that $S_s \CSArrow \ldots \CSArrow S_w$ contains 
no application of the rule Slack-Intro.
By Lemma \ref{lemma:resolvent conflict}, an application of Resolve-Cooper to the 
conflicting core $(x,C')$ prevents any further applications of Resolve-Cooper to the same core.
By Def. \ref{definition:conflicting cores}, the constraints learned through an application 
of Resolve-Cooper contain only variables $y$ such that $y \prec x$.
Therefore, an application of Resolve-Cooper blocks a conflicting core $(x,C')$ and 
introduces potential conflicting cores only for smaller variables than $x$.
This strict decrease in the conflicting variables guarantees that we encounter 
only finitely many conflicting cores in unguarded variables.
Therefore, Resolve-Cooper is at most finitely often applicable.
An analogous argument applies to the rules Solve-Div-Left and Solve-Div-Right.
Thus the rules Resolve-Cooper, Solve-Div-Left and Solve-Div-Right are at most finitely often applicable. 

\begin{lemma}
  \label{lemma:saturation}
  Let $\SearchSt{\emptyseq}{C_0} = S_0 \CSArrow S_1 \CSArrow \ldots$ be a sequence of rule applications applied to a problem $C_0$,
following the strictly-two-layered strategy.
  Then the sequence reaches a state $S_w$ after 
finitely many rule applications $S_0 \CSArrow \ldots \CSArrow S_w$ such that there is 
no further application of the rules Resolve-Cooper, Solve-Div-Left, and Solve-Div-Right 
after state $S_w$. 
\end{lemma}
\begin{proof}
  By Lemma \ref{lemma:prepphase}, we assume w.l.o.g. that the sequence continues from a state $S_s$ such that $S_s$ is reached by the sequence after at most finitely many rule applications $S_0 \CSArrow \ldots \CSArrow S_s$, and there is no further application of the rules Slack-Intro and Forget after state $S_s$. 
  Let $x_1 \prec \ldots \prec x_n$ be the order of variables for all unguarded variables $x_i$.
  We consider a weight vector that strictly decreases after every call to Resolve-Cooper, Solve-Div-Left and Solve-Div-Right.
  For this weight vector, we define $\cores(x_i,C)$ as the set of potential conflicting cores in the problem $C$ with conflicting variable $x_i$.
  Its subset $\unblocked(x_i,C) \subseteq \cores(x_i,C)$ is defined so it contains all its potential conflicting cores without a strong resolvent $R \subseteq C$.
  It is easy to see that $|\cores(x_i,C)| \leq 2^{|C|}$ and therefore both functions define finite sets.
  Now we define the weight vector $\cweight(S)$ for every state $S = \SearchSt{M}{C}(\vdash I)$:\newline
  \centerline{$
    \cweight(S)=(|\cores(x_n,C)|,|\unblocked(x_n,C)|,\ldots,|\cores(x_1,C)|,|\unblocked(x_1,C)|)
  $}
  By Definition~\ref{definition:strictly-two-layered}, Conflict(-Div) is only applicable to guarded constraints
  and guarded variables are only propagated through guarded constraints. 
  Therefore, the conflict $I$ in a state $\ConflictSt{M}{C}{I}$ stays always guarded, even after an application of the Resolve rule,
  and Learn is only applicable to guarded constraints.
  Therefore, Resolve-Cooper, Solve-Div-Left, and Solve-Div-Right are the only rules learning potentially unguarded constraints and thereby the only rules that can increase $|\cores(x_i,C)|$ and $|\unblocked(x_i,C)|$ between two subsequent states $S_i \CSArrow S_{i+1}$.
  After all other transitions $S_i \CSArrow S_{i+1}$, it holds that $\cweight(S_i) \geq_{\lexhelp} \cweight(S_{i+1})$.
  Whenever CUTSAT++ applies Solve-Div-Left, Solve-Div-Right or Resolve-Cooper, the weight vector strictly decreases, i.e., $\cweight(S') >_{\lexhelp} \cweight(S)$:
  \begin{enumerate}
  \item By Lemma~\ref{lemma:resolvent conflict}, an application of Resolve-Cooper to conflicting core $(x_i,C^*)$ implies that there is no strong resolvent $R' \subseteq C'$ for $(x_i,C^*)$.
    By Lemma~\ref{lemma:strong resolvent}, the new problem $C = C' \cup R$ contains a strong resolvent $R$ for $(x_i,C^*)$.
    Therefore, $|\unblocked(x_i,C)| < |\unblocked(x_i,C')|$.
    By Definition~\ref{definition:resolvent}, it holds for all $y \in \vars(R)$ that $y \prec x$.
    Thence, Resolve-Cooper has not introduced new potential conflicting cores $(x_j,C^{**})$ with $j \geq i$ and $|\cores(x_j,C)| \leq |\cores(x_j,C')|$ for all $j \geq i$.
    By Lemma~\ref{lemma:replacing resolvents}, $|\unblocked(x_j,C)| \leq |\unblocked(x_j,C')|$ for all $j > i$.
    Therefore, the weight decreases after an application of Resolve-Cooper, i.e., $\cweight(S') >_{\lexhelp} \cweight(S).$
  \item Let Solve-Div-Left (Solve-Div-Right) be applied to the pair of divisibility constraints $(I_1,I_2)$ such that $\topop(I_1) = x_i$ and $\divsolve(x_i,I_1,I_2)=(I'_1,I'_2)$.
    The new constraint set is $C = C' \setminus \{I_1,I_2\} \cup \{I'_1,I'_2\}$.
    The number of potential conflicting cores containing the same divisibility constraint $I = d \mid a x + p \in C''$ in problem $\hat{C} \uplus \{ I \}$ is the same for all divisibility constraints with $top(I) = x$.
    This means, removing $I_1$ and replacing it with $I'_1$ doesn't increase the number of cores, i.e., $|\cores(x_i,C' \setminus \{I_1\} \cup \{I'_1\} )| = |\cores(x_i,C')|$.
    However, since $I_2 \in \cores(x_i,C')$ and we replace $I_2$ with $I'_2$ where $\topop(I'_2)\prec x_i$, we will decrease the number of conflicting cores in $x_i$.
    It is easy to see, that we do not introduce any new conflicting cores for $x_j$ with $j > i$. Thus $|\cores(x_j,C)| = |\cores(x_j,C')|$.
    Finally, Lemma~\ref{lemma:replacing resolvents} implies that $|\unblocked(x_j,C)| \leq |\unblocked(x_j,C')|$ for $j > i$.
    Therefore, $\cweight(C') >_{\lexhelp} \cweight(C)$.
  \end{enumerate}
  We deduce that the $\cweight$ vector monotonically decreases if we continue from the before mentioned state $S_s$.
  Since $>$ is a well-founded order, the lexicographic order $>_{\lexhelp}$ is also well-founded.
  The minimum of the weight order is $( \ldots, 0, \ldots )$.
  As $>_{\lexhelp}$ is well-founded, there exists no way to decrease the weight $\cweight(C_s)$ without reaching the minimum $( \ldots, 0, \ldots )$ after finitely many applications of the rules Solve-Div-Left, Solve-Div-Right, or Resolve-Cooper.
  Finally, the $\cweight$ vector cannot decrease below $( \ldots, 0, \ldots )$ so CUTSAT++ is not able to apply Solve-Div-Left, Solve-Div-Right, or Resolve-Cooper after we reach a state $S$ with $\cweight(S) = ( \ldots, 0, \ldots )$.
  We conclude that the rules Solve-Div-Left, Solve-Div-Right, and Resolve-Cooper are at most finitely often applicable.
\end{proof}

Next, the sequence reaches a state $S_b$ ($b \geq w$) after a finite derivation of rule applications $S_w \CSArrow \ldots \CSArrow S_b$ such that for every guarded variable $x$ the bounds remain invariant, i.e., $\lowerop(x,M_b) = \lowerop(x,M_j)$ and $\upper(x,M_b) = \upper(x,M_j)$ for every state $S_j = \SearchSt{M_j}{C_j}(\vdash I_j)$ after $S_b = \SearchSt{M_b}{C_b}(\vdash I_b)$ ($j \geq b$):
The strictly-two-layered strategy guarantees that only bounds of guarded variables influence the propagation of further bounds for guarded variables.
Any rule application involving unguarded variables does not influence the bounds for guarded variables.
A proof for the termination of the solely guarded case was already provided in~\cite{CutChase}.
At this point we know that the sequence after $S_b$ contains no further propagations, decisions, or conflict analysis for the guarded variables. 

\begin{lemma}
  \label{lemma: B-rules}
  Let $\SearchSt{\emptyseq}{C_0} = S_0 \CSArrow S_1 \CSArrow \ldots$ be a sequence of rule applications applied to a problem $C_0$,
  following the strictly-two-layered strategy.
  Then the sequence reaches a state $S_b$ after 
finitely many rule applications $S_0 \CSArrow \ldots \CSArrow S_b$ such that such that for every guarded variable $x$ the bounds remain invariant.
\end{lemma}
\begin{proof}
  This proof is based on the termination proof for CUTSAT on finite problems, i.e., problems without  unguarded variables \cite{CutChase}.
  It uses a weight function that strictly decreases whenever CUTSAT++ changes a bound for a guarded variable and otherwise stays the same.
  By Lemmas \ref{lemma:prepphase} and \ref{lemma:saturation}, we assume w.l.o.g. that the sequence continues from a state $S_w$ such that $S_w$ is reached by the sequence after at most finitely many rule applications $S_0 \CSArrow \ldots \CSArrow S_w$, and there is no further application of the rules Slack-Intro, Forget, Resolve-Cooper, Solve-Div-Left, and Solve-Div-Right after state $S_w$. 
  The $\level_B$ of a state $S = \SearchSt{M}{C}$ is the number of decisions for guarded variables in $M$.
  The maximal prefix of $M$ containing only $j$ decisions for guarded variables is denoted by $\Bsubseq_j(M) = M_j$ the $j$-th guarded subsequence.
  Since CUTSAT++ is relevant, it prefers to propagate simple constraints.
  This allows us to assume w.l.o.g. that $M_0$ contains both a lower and upper bound for all guarded variables $x$.
  The \emph{guarded weight} of the $j$-th $\level_B$ is defined by the function $w_B(M_j)$:\newline
  \centerline{$
  w_B(M_j) = \sum_{x \ishelp \guardedhelp} \left(\upper(x,M) - \lowerop(x,M) \right).
  $}
  The \emph{guarded weight} of a state is the vector:\newline
  \centerline{$
  \weight_B(\SearchSt{M}{C}) = \langle w_B(\Bsubseq_0(M)), \cdots, w_B(\Bsubseq_n(M))\rangle,
  $}
  where $n$ is the number of  guarded variables.
  We order the two $\weight_B$-vectors of two subsequent search-states with the well-founded lexicographic order $>_{\lexhelp}$ based on the well-founded order $>$.
  It is easy to see that the minimum of $\weight_B$ is $( \ldots, 0, \ldots )$ and that any change to a bound of a guarded variable changes the guarded weight $\weight_B$. 
  Furthermore, by the definition of the strictly-two-layered strategy, we see that we only propagate  guarded variables with  guarded constraints.
  Thus the strategy also implies that the conflict rules, Conflict, Conflict-Div, Backjump, Resolve, Skip-Decision, Unsat, and Learn, only handle guarded constraints.
  Given the proof for Theorem 2 in~\cite{CutChase}, we see that every application of Propagate, Propagate-Div and Decide applied to a guarded variable
  decreases $\weight_B$ strictly.
  We see in the same proof~\cite{CutChase} that $\weight_B$ strictly decreases between one application of Conflict(-Div) and Backjump as long as 
  the conflict rules handle only  guarded constraints - as is the case for CUTSAT++. 
  Since the bound sequence $M$ is finite, the conflict rules are at most $|M|$ times applicable between one application of Conflict(-Div), and Backjump or Unsat.
  The remaining rules, Propagate, Propagate-Div and Decide applied to unguarded variables, have no influence on $\weight_B$ or the bounds of guarded variables.
  Since the guarded weight $\weight_B$ cannot decrease below $( \ldots, 0, \ldots )$, we conclude that CUTSAT++ is not able to change the bounds for guarded variables infinitely often. 
\end{proof}

Next, the sequence reaches a state $S_u$ ($u \geq b$) after a finite derivation of rule applications $S_b \CSArrow \ldots \CSArrow S_u$ such that also for every unguarded variable $x$ the bounds remain invariant, i.e. $\lowerop(x,M_b) = \lowerop(x,M_j)$ and $\upper(x,M_b) = \upper(x,M_j)$ for every state $S_j = \SearchSt{M_j}{C_j}(\vdash I_j)$ after $S_u = \SearchSt{M_u}{C_b}(\vdash I_u)$ ($j \geq u$).
After $S_b$, CUTSAT++ propagates and decides only unguarded variables or ends with an application of Sat or Unsat(-Div).
CUTSAT++ employs the strictly-two-layered strategy which is also an eager top-level propagating strategy.
The latter induces a strict order of propagation over the unguarded variables through the top-variable restriction for propagating constraints.
Therefore, any bound for unguarded variable $x$ is influenced only by bounds for variables $y \prec x$.
This strict variable order guarantees that unguarded variables are propagated and decided only finitely often.

\begin{lemma}
  \label{lemma: U-Prop}
  Let $\SearchSt{\emptyseq}{C_0} = S_0 \CSArrow S_1 \CSArrow \ldots$ be a sequence of rule applications applied to a problem $C_0$,
  following the strictly-two-layered strategy.
  Then the sequence reaches a state $S_u$ after 
  finitely many rule applications $S_0 \CSArrow \ldots \CSArrow S_u$
  such that for every unguarded variable $x$ the bounds remain invariant.
\end{lemma}
\begin{proof}
  By Lemmas \ref{lemma:prepphase}, \ref{lemma:saturation}, and \ref{lemma: B-rules}, we assume w.l.o.g. that the sequence continues from a state $S_b = \SearchSt{M_b}{C_b}(\vdash I_b)$ such that $S_b$ is reached by the sequence after at most finitely many rule applications $S_0 \CSArrow \ldots \CSArrow S_b$ and only the rules Sat, Unsat-Div, Propagate, Propagate-Div, and Decide are applied after $S_b$, whereas the last three only for unguarded variables.
  Assume for a contradiction that there exists an infinite CUTSAT++ run starting in $S_b$.
  Since there is only a finite number of unguarded variables and no rule to undo a decision, the Decide rule is applied at most finitely often.
  Furthermore, any application of Sat or Unsat-Div ends a run making it finite.
  This allows us to assume w.l.o.g. that there is no application to the rules Sat, Unsat-Div, and Decide in the infinite run starting in the state $S_b$.
  
  Since there are at most finitely many variables in state $S_b$, and no rule to introduce further variables after $S_b$, there exists a smallest unguarded variable $x$ that is propagated infinitely often.
  We assume w.l.o.g. that the run starting in $S_b$ propagates only variables $y$ bigger than or equal to $x$.
  Therefore, the bounds of all smaller variables $y$, remain invariant in all subsequent states $S_j = \SearchSt{M_j}{C_j}(\vdash I_j)$ of $S_b$,
  i.e., $\lowerop(y,M_j) = \lowerop(y,M_b)$ and $\upper(y,M_j) = \upper(y,M_b)$.
  Since there exists no applicable rule that changes the constraint set,  we notice that the constraint set $C_b$ also remains invariant for all states after $S_b$.
  Thus we find all constraints $C^*$ used to propagate $x$ in the set $C_b$.
  Since CUTSAT++ is eager top-level propagating, any constraint $I \in C^*$ has $x$ as their top variable.
  This leads us to the deduction, that $\bound(I, x,M_j) = \bound(I, x,M_b)$ for all subsequent states $S_j = \SearchSt{M_j}{C_j}$ and inequalities $I \in C^*$.
  Since the bounds defined by the inequalities in $C^*$ remain invariant after state $S_b$, CUTSAT++ propagates $x$ at most finitely often with inequalities.
  
  Therefore, there only exists an infinite CUTSAT++ run if $x$ is propagated infinitely often with Propagate-Div.
  This allows us to assume w.l.o.g. that the run starting in $S_b$ propagates $x$ only with Propagate-Div.
  Next, we deduce that variable $x$ stays unbounded in the remaining states of the derivation sequence.
  Otherwise there exists a finite set $x \in \{l_x, \ldots, u_x\}$ bounding $x$, therefore allowing only finitely many propagations.
  In the case that $x$ stays unbounded, the definition of the eager top-level propagating strategy states that Propagate-Div is only applicable to $x$ if $I_d = d \mid ax + p \in C^*$
  is the only divisibility constraint in $C_b$ with $x$ as their top variable.
  Furthermore, we know, because of Definition~\ref{def:eager top} and Lemma~\ref{lemma: cooper-div equivalence}, that there must exist $v \in \mathbb{Z}$ such that $d \mid a v + \lowerop(p,M_k)$ is satisfied.
  Now if we consider Lemma~\ref{div-prop-lemma} and the conditions $c = \bound(D,x,M), c \leq \upper(x,M)$ and $c = \bound(D,x,M), c \geq \lowerop(x,M)$ of the Propagate-Div rule, then we see that Propagate-Div propagates $x$ at most finitely often.
  More specifically, if the lower bound of $x$ is $\lowerop(x,M_b) = l_x \neq -\infty$ then Propagate-Div propagates for $x$ at most $v - l_x$ lower bounds.
  If the upper bound of $x$ is $\upper(x,M_b) = u_x \neq \infty$ then Propagate-Div propagates for $x$ at most $u_x - v$ upper bounds.
  This contradicts the assumption that $x$ is the smallest variable propagated infinitely often, which in turn contradicts our initial assumption!
\end{proof}

After state $S_u$, only the rules Sat, Unsat, and Unsat-Div are applicable, which lead to a final state. 
Hence, the sequence $S_0 \CSArrow S_1 \CSArrow \ldots$ is finite.
We conclude that CUTSAT++ always terminates:

\begin{theorem}
  \label{theorem: unbounded termination}
  If CUTSAT++ starts from an initial state $\SearchSt{\emptyseq}{C_0}$, then there is no infinite derivation sequence.
\end{theorem}
\begin{proof}  
  By Lemmas \ref{lemma:prepphase}, \ref{lemma:saturation}, \ref{lemma: B-rules}, and \ref{lemma: U-Prop}, CUTSAT++ reaches a state $S_u$ after which only the rules Sat, Unsat, and Unsat-Div are applicable, which lead to a final state. 
  Therefore, CUTSAT++ does not diverge.
\end{proof}

\subsection{Frozen States}

Our CUTSAT++ algorithm never reaches a frozen state.
Let $x$ be the smallest unfixed variable with respect to $\prec$.
Whenever $x$ is guarded we can propagate a constraint $\pm x - b \leq 0 \in C$ and then fix $x$ by introducing a decision.
If we cannot propagate any bound for $x$, then $x$ is unguarded and stuck, and therefore Slack-Intro is applicable.
If we cannot fix $x$ by introducing a decision, then $x$ is unguarded and there is a conflict constraint.
Guarded conflict constraints are resolved via the Conflict(-Div) rules.
Unguarded conflict constraints are resolved via the strong conflict resolution rules.
Unless a final state is reached CUTSAT++ has always a rule applicable.

The proof outlined above works because CUTSAT++ encounters only unguarded conflict constraints that are either the result of multiple contradicting divisibility constraints resolvable by Solve-Div-Left and Solve-Div-Right, or expressible via a conflicting core.
Since conflicting cores are only defined over constraints and propagated bounds, we have to guarantee that CUTSAT++ never encounters an unguarded conflict constraint $I$ where $x = \topop(I)$ is fixed with a Decision.
We express this property with the following invariant fulfilled by every state visited by CUTSAT++:

\begin{definition}
  \label{definition:eager top-level propagated}
  A state $S = \SearchSt{M}{C}(\vdash I)$ is called \emph{eager top-level propagated} if 
for all unguarded variables $x$, all decisions $\gamma = \decBnd{x}{\bowtie}{b}$ in 
$M = \concseq{ M', \gamma, M''}$, and all constraints $J \in C$ with $\topop(J)=x$: 
(1) all other variables contained in $J$ are 
fixed in $M'$ and (2) $J$ is no conflict in $S$.
\end{definition}

\begin{lemma}
  \label{lemma: all states eager top-level propagated}
  If $S'$ is an eager top-level propagated state (Def.~\ref{definition:eager top-level propagated}), then any successor state $S = \SearchSt{M}{C}(\vdash I)$ reachable by CUTSAT++ is eager top-level propagated.
\end{lemma}

\begin{proof}
Let $S'$ be an eager top-level propagated state and $S$ its successor, i.e., $S' \CSArrow S$.
We prove this Lemma with a case distinction on the rule leading to the above transition:
\begin{enumerate}
\item Let the applied rule be Propagate(-Div). Then $S' = \SearchSt{M'}{C'}$ and $S = \SearchSt{\concseq{M', \propBnd{x}{\bowtie}{b}{J}}}{C'}$.
      Let $J' \in C'$ be the constraint used for propagation, i.e., $J'$ fulfils the properties $\improves(J',x ,M')$, 
      $\bound(J' ,x ,M') = b$ and $J = \tight(J',x ,M')$ (or $J = \divderive(J',x,M')$).
      Let the  unguarded variable $y$ be fixed by a decided bound $\gamma$, i.e., $M' = \concseq{ M'', \gamma, M'''}$.
      Let $I \in C'$ be a constraint with $\topop(I)=y$. 
      Since $S'$ is eager top-level propagated, all variables in $I$ are fixed in $M'$ and $M''$.
      The variable $x$ is not fixed in $M'$, because the predicate $\improves(J',x,M')$ must be true for Propagate(-Div) to be applicable.
      Therefore, $x$ is not contained in $I$ (or $\topop(I)=y \prec x$) and $I$ is still no conflict in $S$.
      Furthermore, all variables in $I$ are still fixed in $\concseq{M', \propBnd{x}{\bowtie}{b}{J}}$.
      We conclude that $S$ is eager top-level propagated.
\item Let the applied rule be Decide. Then $S' = \SearchSt{M'}{C'}$ and $S = \SearchSt{\concseq{M', \decBnd{x}{\bowtie}{b}}}{C'}$.
      We will use the eager top-level propagating strategy (Def.~\ref{def:eager top}) to prove that $S$ is an eager top-level propagated successor state.
      We consider all  unguarded variables $y$ decided in $S'$ by a decided bound $\gamma$.
      Let $I \in C'$ be a constraint with $\topop(I)=y$.
      The bound $\gamma$ is part of $M'$, i.e., $M' = \concseq{ M'', \gamma, M'''}$.
      As $S'$ is eager top-level propagated, all other variables contained in $I$ are fixed in $M'$ and $M''$.
      Since $\lowerop(x, M') < \upper(x, M')$ is a condition of the Decide rule, the variable $x$ is not fixed in $M'$.
      Therefore, $x$ is not contained in $I$ ($\topop(I)=y \prec x$), and $I$ is still no conflict in $S$.
      Furthermore, all variables in $I$ are still fixed in $\concseq{M', \propBnd{x}{\bowtie}{b}{J}}$.
      Next, we prove that the newly decided variable $x$ does not violate that $S$ is eager top-level propagated.
      Considering Def.~\ref{def:eager top}.2(a) we see that Def.~\ref{definition:eager top-level propagated}.1 is fulfilled.
      Similarly, Def.~\ref{def:eager top}.2(b) enforces Def.~\ref{definition:eager top-level propagated}.2.
      We conclude that $S$ is eager top-level propagated.
\item Let the applied rule be Unsat(-Div) or Sat. Then the successor state $S$ is neither a search- or conflict-state. 
      The Lemma is thereby trivially fulfilled.
\item Let the applied rule be Forget. Then $S' = \SearchSt{M'}{C' \cup \{J\}}$ and $S = \SearchSt{M'}{C'}$.
      Therefore, any conflict $I \in C'$ and any decision in $S$ is also contained in $S'$.
      We conclude that $S$ is eager top-level propagated.
\item Let the applied rule be Slack-Intro. Then $S' = \SearchSt{M'}{C'}$, $x$ is a stuck variable in $S'$ and $S = \SearchSt{M'}{C' \cup \{-x_S \leq 0, x - x_S \leq 0, -x - x_S \leq 0\}}$.
      Either Slack-Intro was applied before and $-x_S \leq 0 \in C'$ or $x_S$ has $\upper(x_S,M) = \infty$, and $-x_S \leq 0$ is not a conflict in $S$.
      Since $x$ was stuck in $S'$, it is unfixed, and the top variable in the new constraints $\{x - x_S \leq 0, -x - x_S \leq 0\}$.
      We conclude that $S$ is eager top-level propagated.
\item Let the applied rule be Resolve-Cooper. Then $S' = \SearchSt{M'}{C'}$ and $S = \SearchSt{M}{C' \cup R_c \cup R_k}$.
      Notice that $M = \prefix(M', y)$ with $y = \min_{I \in R_c} \{\topop(I)\}$.
      Therefore, $M$ is the prefix of $M'$ without decisions in variables greater or equal to $y$.
      Since $y \preceq x$ for all $I \in R$ and $x = \topop(I)$, we deduce that any $I \in R$ that is a conflict has no decision for its top variable $x$ in $S$.
      Since $M$ is a prefix of $M'$, every conflict $I \in C'$ appearing in state $S$ also appears in state $S'$.
      Now it is easy to see that $S$ is eager top-level propagated because $S'$ was eager top-level propagated.
\item Let the applied rule be Solve-Div-Right. Then $S' = \SearchSt{M'}{C' \cup \{I_1, I_2\}}$ and $S = \SearchSt{M}{C' \cup  \{I'_1, I'_2\}}$.
      We notice that $M = \prefix(M', y)$ with $y = \topop(I'_2)$.
      Therefore, $M$ is the prefix of $M'$ without decisions in variables greater or equal to $y$, which includes especially the variable $x = \topop(I_1)$.
      We deduce that neither the top variable of $I'_1$ or $I'_2$ is fixed by a decision.
      Since $M$ is a prefix of $M'$, every conflict $I \in C'$ appearing in state $S$ also appears in state $S'$.
      Now it is easy to see that $S$ is eager top-level propagated because $S'$ was eager top-level propagated.
\item Let the applied rule be Solve-Div-Left. Then $S' = \SearchSt{M'}{C' \cup \{I_1, I_2\}}$ and $S = \SearchSt{M'}{C' \cup  \{I'_1, I'_2\}}$.
      Since the bound sequence is the same in both states, every conflict $I \in C'$ appearing in state $S$ also appears in state $S'$.
      By the definition of the Solve-Div-Left rule, $I'_2$ is no conflict in state $S$.
      Note that $\divsolve$ is an equivalence preserving transformation.
      Thus if $I'_1$ were a conflict in $S$, and $\topop(I'_1)=x$ fixed by a Decision, then $I_1$ or $I_2$ is a conflict in $S'$.
      Therefore, $I'_1$ is no conflict or $\topop(I'_1) = x$ is not decided by a Decision.
      Now it is easy to see that $S$ is eager top-level propagated because $S'$ was eager top-level propagated.
\item Let the applied rule be Conflict or Conflict-Div. Then $S' = \SearchSt{M'}{C'}$ and $S = \ConflictSt{M'}{C'}{I}$.
      It is easy to see that $S$ is eager top-level propagated because $S'$ is eager top-level propagated.
\item Let the applied rule be Resolve or Skip-Decision. Then $S' = \ConflictSt{\concseq{M, \gamma}}{C'}{J'}$ and $S = \ConflictSt{M}{C'}{J}$.
      Since $M$ is a prefix of $M'$, every conflict $I \in C'$ appearing in state $S$ also appears in state $S'$.
      Now it is easy to see that $S$ is eager top-level propagated because $S'$ was eager top-level propagated.
\item Let the applied rule be Learn. Then $S' = \ConflictSt{\concseq{M', \gamma}}{C'}{I}$ and $S = \ConflictSt{M'}{C' \cup {I}}{I}$.
      Since CUTSAT++ uses a two-layered strategy (Def.~\ref{definition:strictly-two-layered}), $I$ is a guarded constraint.
      Now it is easy to see that $S$ is eager top-level propagated because $S'$ was eager top-level propagated.
\item Let the applied rule be Backjump. Then $S' = \ConflictSt{\concseq{M', \gamma, M''}}{C'}{I}$ and $S = \SearchSt{\concseq{M', \gamma'}}{C'}$.
      Since CUTSAT++ uses a two-layered strategy (Def.~\ref{definition:strictly-two-layered}), $I$ is a guarded constraint.
      Now it is easy to see that $S$ is eager top-level propagated because $S'$ was eager top-level propagated.
\end{enumerate}
\end{proof}

Since the initial state $\SearchSt{\emptyseq}{C_0}$ fulfils the eager top-level propagated properties trivially, it is clear that CUTSAT++ produces only eager top-level states, except for the final states.
The eager top-level propagated property is so important, because it allows us to show that CUTSAT++ resolves any conflict $I$ it encounters.
In case the conflict is a guarded constraint this is done with the conflict rules.
Otherwise, the conflict $I$ is an unguarded constraint and CUTSAT++ simulates weak Cooper elimination with the strong conflict resolution rules.
First, we use Solve-Div-Left to simulate Phase I. This either ends with a call to Solve-Div-Right resolving the conflict or CUTSAT++ finds a conflicting core. Then the conflicting core is resolved with the rules Resolve-Cooper.

\begin{lemma}
  \label{lemma:conflict resolvable}
  Let $S = \SearchSt{M}{C}$ be a state reachable by CUTSAT++.
  Let $I \in C$ be a conflict in state $S$. Then state $S$ is not frozen.
\end{lemma}
\begin{proof}
  Assume for a contradiction that state $S$ is frozen.
  W.l.o.g. we assume that $x = \topop(I)$ is the smallest variable in our order that is top variable in a conflict $I' \in C$.
  If $x$ is a guarded variable then Conflict or Conflict-Div is applicable, which contradicts our initial assumption!
  Therefore, $x$ is an unguarded variable.
  Furthermore, all variables $y$ smaller than $x$ are fixed.
  Otherwise, we deduce for the smallest unfixed variable $y$ that either
  \begin{itemize}
  \item $y$ is stuck and Slack-Intro is applicable
  \item Propagate is applicable to a constraint $I'$ where $\topop(I') = y$
  \item $C$ contains at least two divisibility constraints $I_1, I_2$ that have $y$ as their top variable and Solve-Div-Left or Solve-Div-Right is applicable
  \item $S$ contains a diophantine conflicting core $(y,I_d)$ and Resolve-Cooper is applicable
  \item Decision is applicable to $y$ because all conditions in Def.~\ref{def:eager top}.2 are fulfilled
  \end{itemize}
  Since $S$ is eager top-level propagated and $I$ is a conflict with top variable $x$, we know that state $S$ contains no decision for $x$ (Def.~\ref{definition:eager top-level propagated} and Lemma~\ref{lemma: all states eager top-level propagated}).
  W.l.o.g. we assume that $C$ contains at most one divisibility constraint $I_d$ with $x$ as its top variable.
  Otherwise, Solve-Div-Left or Solve-Div-Right are applicable, which contradicts our initial assumption!
  Let $x \geq b_l$ be the strictest lower bound $b_l = \bound(x,I_l,M,{\geq})$ for an inequality $I_l \in C$ with top variable $x$
  or $-\infty$ if there is no inequality propagating a lower bound.
  Let $x \leq b_u$ be the strictest upper bound $b_u = \bound(x,I_u,M,{\leq})$ for an inequality $I_u \in C$ with top variable $x$
  or $\infty$ if there is no inequality propagating an upper bound.
  Since the strictly-two-layered strategy forbids the application of Forget to unguarded constraints,
  CUTSAT++ never removes an unguarded inequality.
  Furthermore, any bound $x \bowtie b$ propagated from a divisibility constraint requires another bound $x \bowtie b'$ propagated from an inequality.
  We deduce that $b_u \neq \infty$ if $\upper(x,M) \neq \infty$ and $b_l \neq -\infty$ if $\lowerop(x,M) \neq -\infty$.
  Next, we do a case distinction on whether the bounds $b_u$ and $b_l$ are finite:
  \begin{itemize}
    \item Let $b_u = \infty$ and $b_l = -\infty$. Then it holds for all inequalities $a x + p \leq 0$ that $\lowerop(a x + p) = -\infty$.
          Thus $I$ is no inequality.
          A divisibility constraint is only a conflict if $\lowerop(x,M) \neq -\infty$ and $\upper(x,M) \neq \infty$.
          This contradicts the assumption that $I$ is a conflict.
    \item Let $b_u = \infty$ and $b_l \in \mathbb{Z}$. Then it holds for all inequalities $a x + p \leq 0$ with $a < 0$ that $\lowerop(a x + p) = -\infty$.
          Thus $I$ is no inequality.
          A divisibility constraint is only a conflict if $\lowerop(x,M) \neq -\infty$ and $\upper(x,M) \neq \infty$.
          This contradicts the assumption that $I$ is a conflict.
    \item Let $b_l = -\infty$ and $b_u \in \mathbb{Z}$. Then it holds for all inequalities $a x + p \leq 0$ with $a > 0$ that $\lowerop(a x + p) = -\infty$.
          Thus $I$ is no inequality.
          A divisibility constraint is only a conflict if $\lowerop(x,M) \neq -\infty$ and $\upper(x,M) \neq \infty$.
          This contradicts the assumption that $I$ is a conflict.
    \item Let $b_u < b_l$. Then $(x,\{I_l,I_u\})$ is a conflicting core, and Resolve-Cooper is applicable.
          This contradicts the assumption that no rule is applicable.
    \item Let $\{b_l, \ldots, b_u\}\neq \emptyset$. Then $I$ is not an inequality.
          If $(x,\{I_l,I_u,I_d\})$ is a conflicting core, then Resolve-Cooper is applicable contradicting our initial assumption.
          Therefore, there exists a solution $b_d \in \{b_l, \ldots, b_u\}$ for $x$ satisfying $I_d$.
          Let $D$ be the set of divisibility constraints used to propagate a bound for $x$ in $M$. 
          All constraints $D' \subseteq D$ not contained in $C$, i.e., $D' = D \setminus C = D \setminus \{I_d\}$ were eliminated with $\divsolve$.
          It is easy to see that there exists a set of constraints $D^* = D^{**} \cup \{I_d\}$ contained in $C$ that implies satisfiability of $D$:\newline
          \centerline{$
            D^* = D^{**} \cup \{I_d\} \rightarrow D,
          $}
          and $D^{**}$ contains only variables $y$ smaller than $x$ (Proof the same as for Lemma~\ref{lemma:replacing resolvents}).
          In state $S$, the set of divisibility constraints $D^*$ is fixed, and satisfied under the partial assignment of $M$.
          Otherwise $S$ would contain a conflict $I' \in D^* \subseteq C$ with $\topop(I') \prec x$.
          This implies that the solution $b_d \in \{b_l, \ldots, b_u\}$ for $x$ that satisfies $I_d$ in $M$, also satisfies $D \cup \{I_d\}$.
          Furthermore, all propagated constraints are satisfied if $x$ is set to $b_d$:\newline
          \centerline{$
             \lowerop(x,M) \leq b_d \leq \upper(x,M).
          $}
          This contradicts the assumption that there exists a conflict $I$ with $\topop(I)=x$.
  \end{itemize}
\end{proof}

The remainder of the proof follows directly the proof outline from above:

\begin{theorem}
  \label{th:frozen}
  Let $S = \SearchSt{M}{C}$ be a state reachable by CUTSAT++.
  Then $S$ is not frozen.
\end{theorem}
\begin{proof}
    Assume for a contradiction that $S = \SearchSt{M}{C}$ is a frozen state.
    It is easy to see that CUTSAT++ can propagate at least two bounds for every guarded variable and afterwards use a Decision to fix them.
    Therefore, we assume that all guarded variables are fixed.
    By Lemma~\ref{lemma:conflict resolvable}, there is no conflict in state $S$.
    Since there is no conflict, at least one variable is unfixed or rule Sat~\cite{CutChase} would be applicable.
    Therefore, there must exist a smallest unfixed and unguarded variable $x$.  
    With the Slack-Intro-rules CUTSAT++ introduces for all variables at least one lower or upper bound.
    Therefore, there exists a violation to the conditions in Def.~\ref{def:eager top}.2 or Decide would be applicable to $x$.
    Since $x$ is the smallest unfixed variable, condition Def.~\ref{def:eager top}.2(a) holds.
    Def.~\ref{def:eager top}.2(c) is also easy to satisfy by applications of Solve-Div-Left or Solve-Div-Right.
    Therefore, Def.~\ref{def:eager top}.2(b) is violated.
    This implies that there exists a constraint $I \in C$ that is a conflict in $S' = \SearchSt{\concseq{ M, \gamma }}{C}$, where $\gamma$ is a decision in $x$ and $x = \topop(I)$.
    However, by Lemma~\ref{lemma:conflict resolvable}, it is not possible that $I \in C$ is a conflict in $S$ or $S$ would not be frozen.
    Finally, $I$ is a conflict only in $S'$ and not $S$ if Propagate(-Div) is applicable to $I$.
    With Solve-Div-Left and Solve-Div-Right it is relatively easy to fulfil the conditions for Def.~\ref{def:eager top}.1 and therefore Propagate(-Div) is applicable.
    We conclude that CUTSAT++ has always one applicable rule, which is a contradiction to our assumption!
\end{proof} 

\subsection{Proof for Completeness}
\label{SSE: Completeness}

All CUTSAT++ rules are sound, i.e., if $\SearchSt{M_i}{C_i}(\vdash I_i) \CSArrow \SearchSt{M_j}{C_j}(\vdash I_j)$ then any satisfiable assignment $\upsilon$ for $C_j$ is a satisfiable assignment also for $C_i$.
The rule Resolve-Cooper is sound because of the Lemmas \ref{lemma: cooper-wc equivalence} and \ref{lemma: cooper-div equivalence}.
The soundness of Solve-Div-Left and Solve-Div-Right follows from the fact that $\divsolve$ is an equivalence preserving transformation.
The soundness proofs for all other rules are either trivial or given in~\cite{CutChase}.

Summarizing, CUTSAT++ is terminating, sound, and never reaches a frozen state.
This in combination with obvious properties of the rules Sat, Unsat, and Unsat-Div implies completeness:

\begin{theorem}
  \label{theorem: unbounded completeness}
  If CUTSAT++ starts from an initial state $\SearchSt{\emptyseq}{C_0}$ then
  it either terminates in the \emph{unsat} state and $C_0$ is unsatisfiable, or it terminates with $\langle \upsilon , \mbox{\emph{sat}} \rangle$ where $\upsilon$ is a satisfiable assignment for $C_0$.
\end{theorem}
\begin{proof}
   By Theorem \ref{theorem: unbounded termination}, CUTSAT++ is terminating.
   By \cite{CutChase}, and the Lemmas \ref{lemma: cooper-wc equivalence} and \ref{lemma: cooper-div equivalence}, CUTSAT++ is sound.
   By Theorem \ref{th:frozen}, CUTSAT++ never reaches a frozen state.
   Since CUTSAT++ is terminating and never reaches a frozen state, every application of CUTSAT++ ends via the rules Sat, Unsat, or Unsat-Div in one of the final states.
   The rule Sat is only applicable in a state $\SearchSt{M}{C}$ where $\upsilon[M ]$ satisfies $C$ and because of soundness also $C_0$.
   The rules Unsat and Unsat-Div are only applicable to states $\SearchSt{M}{C}(\vdash I)$ where the constraint set $C$ contains a trivially unsatisfiable constraint.
   In this case, it follows from the soundness of the CUTSAT++ rules that $C_0$ is unsatisfiable.
\end{proof}

  \section{Conclusion and Future Work}
\label{SE: Conclusion and Future Work}

The starting point of our work was an implementation of CUTSAT~\cite{CutChase} as a theory solver for hierarchic superposition~\cite{FietzkeWeidenbach12}.
In that course, we observed divergence for some of our problems.
The analysis of those divergences led to the development of the CUTSAT++ algorithm presented in this paper, which is a substantial extension of CUTSAT by means of the weak Cooper elimination describe in Section \ref{SE: Weak-Cooper-Elimination}.

As a next step, we plan to develop a prototypical implementation of CUTSAT++, to test its efficiency on benchmark problems.
Depending on the outcome, we consider integrating CUTSAT++ as a theory solver for hierarchic superposition modulo linear integer arithmetic~\cite{FietzkeWeidenbach12}.

Finally, we point at some possible improvements of CUTSAT++.
We see great potential in the development of constraint reduction techniques from (weak) Cooper elimination~\cite{Cooper:72a}.
For practical applicability such reduction techniques might be crucial.
The choice of the variable order $\prec$ has considerable impact on the efficiency of CUTSAT++.
It might be possibly to derive suitable orders via the analysis of the problem structure.
We might benefit from results and experiences of research in quantifier elimination with variable elimination orders.

%  \section*{Acknowledgments}
%  This research was supported in part by the German Transregional Collaborative
%  Research Center SFB/TR 14 AVACS and by the ANR/DFG Programme Blanc project
%  STU~482/2-1 SMArT.
  
\bibliographystyle{plain}
%\nocite{*}
\bibliography{paper-bib}

\begin{thebibliography}{10}

\bibitem{SplittingOnDemand}
Clark Barrett, Robert Nieuwenhuis, Albert Oliveras, and Cesare Tinelli.
\newblock Splitting on demand in sat modulo theories.
\newblock In Miki Hermann and Andrei Voronkov, editors, {\em Logic for
  Programming, Artificial Intelligence, and Reasoning}, volume 4246 of {\em
  Lecture Notes in Computer Science}, pages 512--526. Springer Berlin
  Heidelberg, 2006.

\bibitem{Berman:77a}
Leonard Berman.
\newblock Precise bounds for {P}resburger arithmetic and the reals with
  addition.
\newblock In {\em 18th Annual Symposium on Foundations of Computer Science
  (FOCS 1977), 31 October--2 November, Providence, RI, USA}, pages 95--99.
  IEEE, 1977.

\bibitem{Berman:80a}
Leonard Berman.
\newblock The complexity of logical theories.
\newblock {\em Theoretical Computer Science}, 11(1):71 -- 77, 1980.

\bibitem{BrombergerSW15}
Martin Bromberger, Thomas Sturm, and Christoph Weidenbach.
\newblock Linear integer arithmetic revisited.
\newblock In Amy~P. Felty and Aart Middeldorp, editors, {\em Automated
  Deduction - {CADE-25} - 25th International Conference on Automated Deduction,
  Berlin, Germany, August 1-7, 2015, Proceedings}, volume 9195 of {\em Lecture
  Notes in Computer Science}, pages 623--637. Springer, 2015.

\bibitem{Cooper:72a}
D.~C. Cooper.
\newblock Theorem proving in arithmetic without multiplication.
\newblock In Bernhard Meltzer and Donald Michie, editors, {\em Proceedings of
  the Seventh Annual Machine Intelligence Workshop, Edinburgh, 1971}, volume~7
  of {\em Machine Intelligence}, pages 91--99. Edinburgh University Press,
  1972.

\bibitem{CutProofs}
Isil Dillig, Thomas Dillig, and Alex Aiken.
\newblock Cuts from proofs: A complete and practical technique for solving
  linear inequalities over integers.
\newblock In Ahmed Bouajjani and Oded Maler, editors, {\em Computer Aided
  Verification}, volume 5643 of {\em Lecture Notes in Computer Science}, pages
  233--247. Springer Berlin Heidelberg, 2009.

\bibitem{FerranteRackoff:75a}
Jeanne Ferrante and Charles~W. Rackoff.
\newblock A decision procedure for the first order theory of real addition with
  order.
\newblock {\em SIAM Journal on Computing}, 4(1):69--76, March 1975.

\bibitem{FerranteRackoff:79a}
Jeanne Ferrante and Charles~W. Rackoff.
\newblock {\em The Computational Complexity of Logical Theories}, volume 718 of
  {\em LNM}.
\newblock Springer, 1979.

\bibitem{FietzkeWeidenbach12}
Arnaud Fietzke and Christoph Weidenbach.
\newblock Superposition as a decision procedure for timed automata.
\newblock {\em Mathematics in Computer Science}, 6(4):409--425, 2012.

\bibitem{FischerRabin:74a}
M.~J. Fischer and M.~Rabin.
\newblock Super-exponential complexity of {P}resburger arithmetic.
\newblock {\em SIAM-AMS Proceedings}, 7:27--41, 1974.

\bibitem{Furer:82a}
Martin F{\"u}rer.
\newblock The complexity of presburger arithmetic with bounded quantifier
  alternation depth.
\newblock {\em Theoretical Computer Science}, 18(1):105 -- 111, 1982.

\bibitem{Gradel:87a}
Erich Gr{\"{a}}del.
\newblock {\em The complexity of subclasses of logical theories}.
\newblock PhD thesis, Universit{\"a}t Basel, 1987.

\bibitem{Mathsat}
Alberto Griggio.
\newblock A practical approach to satisability modulo linear integer
  arithmetic.
\newblock {\em JSAT}, 8(1/2):1--27, 2012.

\bibitem{JovanovicM11}
Dejan Jovanovi\'{c} and Leonardo de~Moura.
\newblock Cutting to the chase solving linear integer arithmetic.
\newblock In {\em Automated Deduction - CADE-23 - 23rd International Conference
  on Automated Deduction, Wroclaw, Poland, July 31 - August 5, 2011.
  Proceedings}, volume 6803 of {\em Lecture Notes in Computer Science}, pages
  338--353. Springer, 2011.

\bibitem{CutChase}
Dejan Jovanovi\'{c} and Leonardo de~Moura.
\newblock Cutting to the chase.
\newblock {\em Journal of Automated Reasoning}, 51(1):79--108, 2013.

\bibitem{JungerLiebling:10a}
Michael J{\"u}nger, Thomas~M. Liebling, Denis Naddef, George~L. Nemhauser,
  William~R. Pulleyblank, Gerhard Reinelt, Giovanni Rinaldi, and Laurence~A.
  Wolsey, editors.
\newblock {\em 50 Years of Integer Programming 1958-2008}.
\newblock Springer, 2010.

\bibitem{LasarukSturm:07a}
Aless Lasaruk and Thomas Sturm.
\newblock Weak quantifier elimination for the full linear theory of the
  integers. {A} uniform generalization of {P}resburger arithmetic.
\newblock {\em Applicable Algebra in Engineering, Communication and Computing},
  18(6):545--574, December 2007.

\bibitem{Oppen:78a}
Derek~C. Oppen.
\newblock A $2^{2^{2^{pn}}}$ upper bound on the complexity of {P}resburger
  arithmetic.
\newblock {\em J. Comput. Syst. Sci.}, 16(3):323--332, 1978.

\bibitem{BoundingBox}
Christos~H. Papadimitriou.
\newblock On the complexity of integer programming.
\newblock {\em J. ACM}, 28(4):765--768, October 1981.

\bibitem{Presburger:29a}
Mojzesz Presburger.
\newblock {\"U}ber die {V}ollst{\"a}ndigkeit eines gewissen {S}ystems der
  {A}rithmetik ganzer {Z}ahlen, in welchem die {A}ddition als einzige
  {O}peration hervortritt.
\newblock In {\em Comptes Rendus du premier congres de Mathematiciens des Pays
  Slaves}, pages 92--101, Warsaw, Poland, 1929.

\bibitem{GathenSieveking:78a}
Joachim von~zur Gathen and Malte Sieveking.
\newblock A bound on solutions of linear integer equalities and inequalities.
\newblock {\em Proceedings of the AMS}, 72:155--158, 1978.

\bibitem{Weispfenning:90a}
Volker Weispfenning.
\newblock The complexity of almost linear diophantine problems.
\newblock {\em Journal of Symbolic Computation}, 10(5):395--403, November 1990.

\end{thebibliography}

\end{document}